%% file: main.tex
\documentclass{llncs}

\usepackage{amsfonts}
\usepackage{xcolor}
\usepackage{graphicx}
\usepackage{amssymb}
\usepackage{amsmath}
\usepackage{stmaryrd}

\usepackage{algorithm}
\usepackage[indLines=false]{algpseudocodex}

\usepackage{hyperref}
\usepackage{cleveref}
\usepackage{xspace}
\usepackage{marginnote}
\usepackage{multirow}
\usepackage{rotating}
\usepackage{colortbl}
\usepackage{thm-restate}
\usepackage{enumitem}
\usepackage{orcidlink} 
\usepackage[misc,geometry]{ifsym} 

\usepackage{float}
\usepackage{caption}
\usepackage{subcaption} 

\usepackage{array}
\usepackage{booktabs}
\usepackage{soul} 
\usepackage{hhline}

\usepackage{adjustbox}
\usepackage{etoolbox}
\AtEndEnvironment{proof}{\qed}
\usetikzlibrary{tikzmark, arrows.meta, positioning}

\crefformat{line}{#2l. #1#3}
\crefrangeformat{line}{ll. #3#1#4--#5#2#6}
\crefmultiformat{line}{ll. #2#1#3}{ and #2#1#3}{, #2#1#3}{ and #2#1#3}

\input{macros_fig_imi}

\input{macros}

\title{State-Space Abstractions for\\Parametric Timed Games
     \thanks{This work was supported by
         the CNRS International Research Network CLoVe
         and Innovationsfonden Danmark's DIREC project SIoT (Secure Internet of Things).
        }
    \thanks{This is the full version of the paper under the same title accepted to
QEST+FORMATS 2026.}
    }

\author{
 Mikael Bisgaard Dahlsen-Jensen \Letter\inst{1}\orcidlink{0000-0003-0641-7635}
 ,
 Laure Petrucci \inst{2}\orcidlink{0000-0003-3154-5268},\\
 Jaco van de Pol \inst{1}\orcidlink{0000-0003-4305-0625}
 }

\institute{
 Aarhus University, Aarhus, Denmark
 \\\email{\{mikael,jaco\}@cs.au.dk}
 \and
 Université Sorbonne Paris Nord, LIPN, CNRS UMR 7030, Villetaneuse, France
 \\ \email{Laure.Petrucci@lipn.univ-paris13.fr}
 }

 \authorrunning{Dahlsen-Jensen, Petrucci, van de Pol}

\titlerunning{State-Space Abstractions for Parametric Timed Games}

\RequirePackage{graphicx}
\graphicspath{{badges/}}
\usepackage[firstpageonly=true,angle=0,vpos=.147\paperheight,hpos=.39\linewidth,vanchor=t,hanchor=l]{draftwatermark}
\SetWatermarkText{%
\raisebox{-1cm}
{\begin{minipage}{\linewidth}\centering\includegraphics[width=12.5mm]
{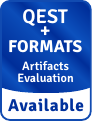}\hfill\includegraphics[width=12.5mm]
{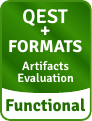}\end{minipage}}%
}

\begin{document}

\maketitle

\setlength{\abovedisplayskip}{0pt}
\setlength{\belowdisplayskip}{0pt}
\begin{abstract}
Synthesizing controllers for real-time systems under both timing uncertainty and adversarial environments requires exploring prohibitively large symbolic state spaces. While zone inclusion checking has been applied to Parametric Timed Games, more aggressive abstractions from the Parametric Timed Automata and Timed Games literature---double inclusion, zone merging, hull abstractions, and location-based abstraction---have not yet been lifted to the parametric game setting. We define a general abstraction framework for Parametric Timed Games, instantiate it with each of the aforementioned abstractions, and prove that the framework preserves correctness of parameter synthesis and winning strategies. Experimental results on an established production cell benchmark and a novel adversarial IoT case study show that the abstractions significantly improve scalability, solving instances previously intractable for existing techniques.
\end{abstract}

\section{Introduction}
\label{sec:intro}
\input{intro}

\section{Preliminaries}
\input{prelim}

\section{State space reduction techniques}
\input{abstractions}

\section{Experimental Evaluation}
\input{eval}

\section{Conclusion}
\input{conclusion}

\subsubsection*{Acknowledgements.}
We thank Baptiste Fievet for contributions to the algorithmic ideas, and Étienne André for advice on adapting zone merging to the parametric game setting.
Parts of this paper used generative AI tools (Claude by Anthropic) to assist with drafting and editing for style and clarity.

\paragraph{Data Availability.} The models, scripts, and tools to reproduce our experimental evaluation are archived and publicly available at~\cite{abstractions_artifact}. 

\bibliographystyle{plain}
\bibliography{bibliography.bib}

\newpage
\appendix
\input{appendix}

\end{document}

%% file: macros_fig_imi.tex
\usepackage{pgfplots}
\usepackage{tikz}

\usetikzlibrary{arrows,automata,fit,positioning,shadows,shapes}
\usetikzlibrary{decorations,snakes,calc} 
\usetikzlibrary{patterns, patterns.meta}
\tikzstyle{location}=[rectangle, rounded corners, minimum size=12pt, draw=black, fill=blue!10, inner sep=2pt]
\tikzstyle{pta}=[auto, ->, >=stealth']
\tikzstyle{PZG}=[auto, ->, >=stealth']
\tikzstyle{mergingFigure} = [>=stealth', node distance=1.8cm, yscale=.6]
\tikzstyle{location10}=[location, minimum size=10pt]
\tikzstyle{invariant}=[draw=black, dotted, inner sep=1pt, node distance=0] 
\tikzstyle{final}=[fill=green!70,double]
\tikzstyle{urgent}=[dotted, draw=red, very thick, fill=yellow]
\tikzstyle{bad}=[fill=red]

\tikzstyle{pzgstate} = [
	rectangle split,
	rectangle split horizontal,
	rectangle split parts=2,
	draw,
	rectangle split part align={center,left},
	rounded corners,
	minimum height=0.5 cm,
	minimum width=1.5 cm,
	thick
	]
\tikzstyle{fillred} = [ fill=red!20 ]
\tikzstyle{fillblue} = [ fill=blue!20 ]
\tikzstyle{fillyellow} = [ fill=yellow!20 ]
\tikzstyle{line} = [ draw,-latex',thick ]
\tikzstyle{highlightarrow} = [
	preaction={
		draw,
		line width=0.5mm
	}
]

\tikzstyle{urgent}=[fill=yellow, thick, dotted] 
\tikzstyle{private}=[fill=red!50,thick]

\pgfdeclarepatternformonly{north west lines wide}
{\pgfqpoint{-1pt}{-1pt}}
{\pgfqpoint{10pt}{10pt}}
{\pgfqpoint{9pt}{9pt}}
{
	\pgfsetlinewidth{3pt}
	\pgfpathmoveto{\pgfqpoint{0pt}{9.1pt}}
	\pgfpathlineto{\pgfqpoint{9.1pt}{0pt}}
	\pgfusepath{stroke}
}

\tikzset{onslide/.code args={<#1>#2}{%
		\only<#1>{\pgfkeysalso{#2}} 
}}

\definecolor{coloract}{rgb}{0.50, 0.70, 0.30}
\definecolor{colorclock}{rgb}{0.4, 0.4, 1}
\definecolor{colordisc}{rgb}{1, 0, 1}
\definecolor{colorloc}{rgb}{0.4, 0.4, 0.65}

\definecolor{colorparam}{rgb}{.66, 0.4, 0.0}
\definecolor{colorstate}{rgb}{1, 0.4, 0.0}

\definecolor{loccolor1}{rgb}{1, 0.6, 0.45}
\definecolor{loccolor2}{rgb}{0.45, 1, 0.45}
\definecolor{loccolor3}{rgb}{0.8, 0.8, 1}
\definecolor{loccolor4}{rgb}{1, 0.45, 1}
\definecolor{loccolor5}{rgb}{1, 1, 0.45}
\definecolor{loccolor6}{rgb}{0.45, 1, 1}
\definecolor{loccolor7}{rgb}{0.9, 0.6, 0.2}
\definecolor{loccolor8}{rgb}{0.7, 0.4, 1}
\definecolor{loccolor9}{rgb}{0.5, 1, 0.75}
\definecolor{loccolor10}{rgb}{0.8, 0.7, 0.6}
\definecolor{loccolor11}{rgb}{0.6, 0.7, 0.8}
\definecolor{loccolor12}{rgb}{0.2, 0.5, 0.9}
\definecolor{loccolor13}{rgb}{0.5, 0.9, 0.2}
\definecolor{loccolor14}{rgb}{0.9, 0.2, 0.5}
\definecolor{loccolor15}{rgb}{0.7, 0.7, 0.7}
\definecolor{loccolor16}{rgb}{0.8, 0.8, 0.5}

\usepackage{xcolor}
\definecolor{darkgreen}{rgb}{0.0, 0.4, 0.08}
\definecolor{lighterblack}{rgb}{.4, .4, .4}
\definecolor{colorparam}{rgb}{1, 0.6, 0.0}
\definecolor{mygreen}{rgb}{0,0.6,0}
\definecolor{mygray}{rgb}{0.5,0.5,0.5}
\definecolor{mymauve}{rgb}{0.58,0,0.82}

\definecolor{gris}{rgb}{0.6,0.6,0.6}
\definecolor{grisfonce}{rgb}{0.2,0.2,0.2}
\definecolor{turquoise}{rgb}{0, 1, 1}
\definecolor{vertfonce}{rgb}{0,0.85,0}
\definecolor{violet}{rgb}{0.8,0,0.8}
\definecolor{grispale}{rgb}{0.9, 0.9, 0.9}
\definecolor{cpale1}{rgb}{1, 0.3, 0.3}
\definecolor{cpale2}{rgb}{0.3, 1, 0.3}
\definecolor{cpale3}{rgb}{0.3, 0.3, 1}
\definecolor{cpale4}{rgb}{1, 0.3, 1}
\definecolor{cpale5}{rgb}{1, 1, 0.3}
\definecolor{cpale6}{rgb}{0.3, 1, 1}
\definecolor{cpale7}{rgb}{0.9, 0.6, 0.2}
\definecolor{cpale8}{rgb}{0.7, 0.4, 1}
\definecolor{cpale9}{rgb}{0.5, 1, 0.75}
\definecolor{cpale10}{rgb}{0.8, 0.7, 0.6}
\definecolor{cpale11}{rgb}{0.6, 0.7, 0.8}
\definecolor{cpale12}{rgb}{0.2, 0.5, 0.9}
\definecolor{cpale13}{rgb}{0.5, 0.9, 0.2}
\definecolor{cpale14}{rgb}{0.9, 0.2, 0.5}
\definecolor{cpale15}{rgb}{0.7, 0.7, 0.7}
\definecolor{cpale16}{rgb}{0.8, 0.8, 0.5}
\definecolor{bleuciel}{rgb}{0.90,0.95,1}
\definecolor{cv1}{rgb}{1, 0, 0}
\definecolor{cv2}{rgb}{0, 1, 0}
\definecolor{cv3}{rgb}{0, 0, 1}
\definecolor{cv4}{rgb}{1, 1, 0}
\definecolor{cv5}{rgb}{1, 0, 1}
\definecolor{cv6}{rgb}{0, 1, 1}
\definecolor{cv7}{rgb}{0.8, 0.6, 0.4}
\definecolor{cv8}{rgb}{0.5, 0.5, 1}
\definecolor{cv9}{rgb}{0.55, 0.75, 0.35}
\definecolor{cv10}{rgb}{1, 0.6, 0.1}
\definecolor{cv11}{rgb}{0.6, 0.7, 0.8}
\definecolor{cv12}{rgb}{0.2, 0.5, 0.9}
\definecolor{cv13}{rgb}{0.5, 0.9, 0.2}
\definecolor{cv14}{rgb}{1, 0.3, 0.5}
\definecolor{cv15}{rgb}{0.7, 0.7, 0.7}
\definecolor{cv16}{rgb}{0.8, 0.8, 0.5}
\definecolor{cvorange}{rgb}{1,.8,0.5}
\definecolor{colortask}{rgb}{0.5, 0.2, 0.9} 

%% file: macros.tex
\newcommand{\marginX}{\marginnote{\huge{\quad\textbf{!}\quad}}}

\newcommand{\timeout}{\texttt{T/O}\xspace}

\ifdefined\VersionWithComments
  \newcommand{\mbdj}[1]{\textcolor{purple}{\marginX{}[\textbf{Mikael}: #1]}}
  \newcommand{\bfi}[1]{\textcolor{orange}{\marginX{}[\textbf{Baptiste}: #1
  ]}}
  \newcommand{\lp}[1]{\textcolor{green!70!black}{\marginX{}[\textbf{Laure}: #1 ]}}
  \newcommand{\jvdp}[1]{\textcolor{blue!40}{\marginX{}[\textbf{Jaco}:
  #1 ]}}

\else
  \newcommand{\mbdj}[1]{}
  \newcommand{\bfi}[1]{}
  \newcommand{\lp}[1]{}
  \newcommand{\jvdp}[1]{}
\fi

\newcommand{\eg}{e.g.\xspace}
\newcommand{\ie}{i.e.\xspace}

\newcommand{\game}{\ensuremath{G} }

\newcommand{\LocSet}{\ensuremath{\mathbb{L}}}
\newcommand{\loc}{\ensuremath{\ell} }

\newcommand{\ClockSet}{\ensuremath{X}}

\newcommand{\ParamSet}{\ensuremath{P}}

\newcommand{\TransSet}{\ensuremath{T} }
\newcommand{\trans}{\ensuremath{t} }

\newcommand{\temptrans}{\to^{\delta}}
\newcommand{\disctrans}{\to^{t}}

\newcommand{\guard}{\ensuremath{g} }

\newcommand{\SubsetSet}[1]{\ensuremath{\mathcal{P}(#1)} }

\newcommand{\LabelSet}{\ensuremath{Act} }

\newcommand{\Inv}{\ensuremath{\mathit{Inv}}}

\newcommand{\val}{\ensuremath{v} }
\newcommand{\ClockValSet}[1]{ \ensuremath{\mathbb{R}_{\geq 0}^{#1}}}
\newcommand{\ParamValSet}[1]{ \ensuremath{\mathbb{Q}_{\geq 0}^{#1}}}
\newcommand{\ValSet}{\ensuremath{V}}

\newcommand{\StateSpace}{\ensuremath{\mathbb{S}}}

\newcommand{\ParamLinearTerm}{\ensuremath{plt} }

\newcommand{\ZoneFormula}{\phi}

\newcommand{\delay}{\ensuremath{\delta} }

\newcommand{\TargetSet}{\ensuremath{{R}} }

\newcommand{\zone}{\ensuremath{Z} }
\newcommand{\ZoneSet}{\ensuremath{\mathcal{Z}} }
\newcommand{\ZonesAt}[1]{\ensuremath{\ZoneSet_{#1}}}

\newcommand{\mergefun}{\ensuremath{\alpha}}
\newcommand{\ftriv}{\ensuremath{\mergefun_{\mathrm{triv}}}}
\newcommand{\finc}{\ensuremath{\mergefun_{\subseteq}^{\leftrightarrow}}}
\newcommand{\floc}{\ensuremath{\mergefun_{\mathrm{loc}}}}
\newcommand{\fhull}{\ensuremath{\mergefun_{\mathrm{ch}}}}
\newcommand{\focthull}{\ensuremath{\mergefun_{\mathrm{oct}}}}
\newcommand{\fboxhull}{\ensuremath{\mergefun_{\mathrm{box}}}}
\newcommand{\fhullgen}{\ensuremath{\mergefun_{h}}}
\newcommand{\repmap}[1]{\ensuremath{\rho_{#1}}}
\newcommand{\PZGred}[1]{\ensuremath{\mathit{PZG}^{#1}}}

\newcommand{\SymbState}{\ensuremath{\xi} }
\newcommand{\SymbStateSet}{\ensuremath{\Xi} }
\newcommand{\SymbStateInit}{\ensuremath{\xi_I} }

\newcommand{\TempPred}[1]{ \ensuremath{ #1^{\swarrow} } }
\newcommand{\TempSucc}[1]{ \ensuremath{ #1^{\nearrow} } }

\newcommand{\ProjectParam}[1] { \ensuremath{#1{\downarrow_P}}}

\newcommand{\imitator}{\textsc{Imitator}\xspace}
\newcommand{\Uppaal}{\textsc{Uppaal}\xspace}

\newcommand{\Win}{\ensuremath{\mathit{Win}}}

\newcommand{\Depends}{\ensuremath{\mathit{Depends}}}
\newcommand{\Succ}{\ensuremath{\mathit{Succ}}}
\newcommand{\Pred}{\ensuremath{\mathit{Pred}}}
\newcommand{\SafePred}{\ensuremath{\mathit{SafePred}}}
\newcommand{\NewWin}{\ensuremath{\mathit{NewWin}}}

\newcommand{\WinningParam}{\ensuremath{\mathit{WinningParam}}}
\newcommand{\Uncontrollable}{\ensuremath{\mathit{UnCtrl}}}

\newcommand{\Explored}{\ensuremath{\mathit{Explored}}}
\newcommand{\WaitingExplore}{\ensuremath{\mathit{WaitingExplore}}}
\newcommand{\WaitingUpdate}{\ensuremath{\mathit{WaitingUpdate}}}
\definecolor{ColorLosingProp}{HTML}{1B9E77}
\definecolor{ColorCumuPrune}{HTML}{D95F02}
\definecolor{ColorCovPrune}{HTML}{7570B3}
\definecolor{ColorInc}{HTML}{E7298A}

\crefname{definition}{Def.}{Defs.}
\crefname{theorem}{Thm.}{Thms.}
\crefname{lemma}{Lem.}{Lemmas}
\crefname{line}{Line}{Lines}
\crefname{algorithm}{Alg.}{Algs.}
\crefname{figure}{Fig.}{Figs.}
\crefname{appendix}{Appendix}{Appendices}
\crefname{section}{Sec.}{Sections}
\crefname{table}{Table}{Tables}

\newcommand{\WinningMove}{\ensuremath{\mathit{WinningMove}}}

%% file: intro.tex

State-space abstractions have proven effective for scaling the verification of parametric real-time systems, but a general abstraction framework is lacking for the game-theoretic setting where a controller must win against an adversary.
Timed automata (TA)~\cite{Intro-TA} provide a formal model for real-time systems, extended by
Parametric Timed Automata (PTA)~\cite{DBLP:conf/stoc/AlurHV93}, which replace concrete timing constants
with unspecified parameters.
PTA support \emph{parameter synthesis}: computing the set of parameter valuations under which a given property holds.
Timed Games (TG)~\cite{Classic-TG,DBLP:conf/concur/CassezDFLL05} complement this by introducing a two-player structure,
distinguishing \emph{controllable} actions (chosen by a controller) from \emph{uncontrollable} ones
(chosen by an adversary), and asking whether a winning strategy exists for the controller.
Parametric Timed Games (PTG)~\cite{DBLP:conf/wodes/JovanovicFLR12,DBLP:journals/ijcon/JovanovicLR19}
combine both settings: they synthesize the set of parameters under which the controller can guarantee winning.
An on-the-fly PTG reachability algorithm was implemented in~\cite{DBLP:conf/tacas/DahlsenJensenFPP24}
and extended to full controller synthesis in~\cite{DBLP:conf/qestformats/DahlsenJensenFPP25}.

PTG algorithms suffer from \emph{state-space explosion}:
the state space is explored symbolically using \emph{parametric zones}, convex constraints over clocks and parameters. Reducing the number of zones while preserving correctness is the central challenge addressed in this paper.
While zone inclusion checking has been applied to PTG~\cite{DBLP:conf/tacas/DahlsenJensenFPP24},
more aggressive abstractions remain unexplored:
double inclusion and zone merging~\cite{DBLP:conf/atva/AndreFS13,DBLP:conf/formats/AndreMPP22} from the PTA setting
and location-based abstraction~\cite{DBLP:conf/concur/JensenLLS25} from the TG setting.
Lifting these to PTG is non-trivial:
the PTA techniques must be shown to preserve winning strategies, not just reachability,
while the TG technique must be extended to handle parameters.

\medskip\noindent\textbf{Contributions.} We make the following contributions:
\begin{enumerate}
    \item A \emph{general abstraction framework} for PTG that preserves correctness of parameter synthesis and winning strategies for any valid instantiation (\cref{sec:statespacered}).
    \item Seven concrete instantiations, recovering known PTA and TG techniques and introducing new abstractions inspired by numerical abstract domains (convex hull, octagonal hull and box hull).
    \item An implementation in the \imitator{} model checker~\cite{DBLP:conf/cav/Andre21}, evaluated on a pro\-duction-cell benchmark~\cite{ProdCellCaseStudy} and a novel adversarial IoT case study, showing reduced state spaces and improved scalability, with several instances going from timeout to terminating in seconds.
\end{enumerate}

\medskip\noindent\textbf{Related work.}
Difference Bound Matrices (DBMs)~\cite{DBLP:conf/avmfss/Dill89} are the standard zone representation for timed automata,
restricting zones to conjunctions of difference-bound constraints;
Parametric DBMs~\cite{DBLP:conf/tacas/HuneRSV01} extend this to the parametric setting. Alternatively, symbolic zones can be represented by convex polyhedra~\cite{DBLP:conf/cav/Andre21}.
DBMs can be viewed as abstractions, providing over-approximations of arbitrary convex polyhedra.
In the non-parametric setting, extrapolation techniques such as LU-bounds~\cite{DBLP:conf/tacas/BehrmannBLP04}
are widely used zone abstractions for timed automata, \eg in \Uppaal{}~\cite{DBLP:journals/sttt/LarsenPY97}.
Zone extrapolation was first adapted to the parametric setting in~\cite{DBLP:conf/rp/AndreLR15,DBLP:conf/sefm/BezdekBBC16} and developed further in~\cite{DBLP:conf/nfm/ArcileA22}.

Zone inclusion checking was shown to be effective for timed games in~\cite{DBLP:conf/concur/CassezDFLL05}
and has been used in PTG~\cite{DBLP:conf/tacas/DahlsenJensenFPP24}; our framework provides a general correctness proof subsuming this technique.
For PTA, zone merging was introduced in~\cite{DBLP:conf/atva/AndreFS13}
and further developed in~\cite{DBLP:conf/formats/AndreMPP22}.
Our hull abstractions draw on classical numerical abstract domains:
intervals~\cite{DBLP:conf/popl/CousotC77},
convex polyhedra~\cite{DBLP:conf/popl/CousotH78},
and octagons~\cite{DBLP:conf/wcre/Mine01}. We now lift these abstractions to PTGs.
A form of convex hull approximation for zones was introduced in~\cite{DBLP:conf/rtss/Balarin96} and is used in \Uppaal{}: it computes the smallest DBM containing two given zones, which includes the true convex hull but may be strictly larger when the hull is not representable as a DBM.
A location-based abstraction was used in~\cite{DBLP:conf/concur/JensenLLS25} to speed up on-the-fly symbolic algorithms with abstractions for timed ATL. We lift this abstraction to parametric timed reachability games.
Foundational frameworks for abstracting games~\cite{DBLP:conf/sas/HenzingerMMR00,DBLP:conf/lics/AlfaroGJ04} abstract the moves of the players, requiring an asymmetric treatment of the two players; the \textsc{Synthia} tool~\cite{DBLP:conf/cav/PeterEM11} synthesizes controllers for (non-parametric) timed games by abstraction refinement.
In contrast, our abstractions coarsen only the explored symbolic state space, keeping the winning computation exact (cf.\ \cref{sec:algorithm}).

\medskip\noindent\textbf{Outline.}
\Cref{sec:prelim} introduces PTG and the necessary background.
\Cref{sec:statespacered} defines the abstraction framework, adapts the on-the-fly PTG algorithm, proves its correctness, and instantiates it with concrete abstractions.
\Cref{sec:expe} presents the experimental evaluation, including the adversarial IoT case study.
\Cref{sec:concl} concludes.

%% file: prelim.tex

\label{sec:prelim}
A Parametric Timed Game (PTG) is a structure based on timed automata (TA). Like classical automata, it has locations connected by discrete transitions. It also has clocks.
Locations are associated a condition on clock valuations (invariant) that must
be satisfied while staying in the location.
An action in a timed automaton is either to take a discrete transition or to let some
time pass.
Discrete transitions have a guard that must be
satisfied in order to take the transition. In a parametric setting, these conditions
use linear terms over clocks and parameters. Parameters are unspecified but constant during a run. A discrete transition has a subset of clocks which are reset when the transition is taken.

    A \emph{clock valuation} is a function $\val_\ClockSet \in \ClockValSet
    {\ClockSet}$ assigning a non-negative real value to each clock.
    A \emph{parameter valuation} $\val_\ParamSet \in \ParamValSet
    {\ParamSet}$ assigns a non-negative rational value to each parameter.
    A \emph{valuation of a game} is a pair $\val = (\val_\ClockSet,\val_\ParamSet)$. 
    The set of all valuations of the game is denoted
    $\ValSet = \ClockValSet{\ClockSet} \times \ParamValSet{\ParamSet}$.
A \emph{linear term} over $\ParamSet$ is a term defined by the following grammar:
$\ParamLinearTerm \; := \; k \; | \; kp \; | \; \ParamLinearTerm + \ParamLinearTerm$
where $k \in \mathbb{Q}$ and $p \in \ParamSet$.
\emph{Zones} allow for capturing a set of valuations in a game.
The set of \emph{parametric zones} $\ZoneSet(\ClockSet,\ParamSet)$ is the set of formulas
defined inductively by the following grammar:
\(\ZoneFormula \; := 
\; \top
\; | \; \ZoneFormula \land \ZoneFormula \; | \; x \sim \ParamLinearTerm 
\; | \; x-y \sim \ParamLinearTerm
\; | \; \ParamLinearTerm \sim 0\)
where $x, y\in\ClockSet$, ${\sim} \in \{ <; \leq; =; \geq;
> \}$ and $\ParamLinearTerm$ is a linear term over $\ParamSet$.

In a two-player timed game, discrete transitions are partitioned
between controllable transitions and uncontrollable (environment) transitions.

\begin{definition}[PTG]
A \emph{Parametric Timed Game} is a tuple of the form
$\game = (\LocSet , \ClockSet, \ParamSet, \LabelSet, \TransSet_c, \TransSet_u, \loc_0, \Inv)$ such that
\begin{itemize}[noitemsep, topsep=0pt]
    \item $\LocSet$, $\ClockSet$, $\ParamSet$, $\LabelSet$ are sets of \emph{locations}, \emph{clocks},
        \emph{parameters}, \emph{transition labels}.
    \item
        We write $\TransSet = \TransSet_c \sqcup \TransSet_u$ for the set of all \emph{transitions}, partitioned into \emph{controllable} ($\TransSet_c$) and \emph{uncontrollable} ($\TransSet_u$) ones, where
        $\TransSet \subseteq \LocSet \times \ZoneSet(\ClockSet,\ParamSet) \times
            \LabelSet \times \SubsetSet{\ClockSet} \times \LocSet$
        is the set of \emph{transitions} of the form
        $(\loc,\guard,a,Y,\loc')$ where:
        $\loc$, $\loc'$ are source
        and target locations, $\guard$ is the guard,      
        $a$ the label, $Y$ the set of clocks to reset.
    
    \item $\loc_0$ is the initial location.
    \item $\Inv \; : \; \LocSet \to \ZoneSet(\ClockSet,\ParamSet)$ associates an
        \emph{invariant} with each location.
\end{itemize}
\end{definition}

Function $\val_\ParamSet$ is naturally extended to linear terms on parameters, by
replacing each parameter in the term with its valuation.
With $\val \models \ZoneFormula$, we denote that
valuation $\val = (\val_\ClockSet, \val_\ParamSet)$ \emph{satisfies}
a zone $\ZoneFormula$.
Zones
can also be seen as a convex set in the space of valuations
by considering those satisfying the condition. 

\subsection{Semantics of Parametric Timed Games}

A \emph{state} of a PTG consists of a location and a valuation of clocks and parameters.
Transitions modify clock valuations by letting time pass or resetting clocks.

    Let $\val = (\val_\ClockSet,\val_\ParamSet)$ be a valuation of the game
    and $\delay \geq 0$ a delay.
    $\forall x\in\ClockSet: (\val_\ClockSet+\delay)(x) =
    \val_\ClockSet(x) + \delay$
and
    $\val+\delay = (\val_\ClockSet +\delay , \val_\ParamSet)$.
    Let
    $Y \subseteq \ClockSet$.
    $\val_\ClockSet[Y:=0]$ is the valuation obtained by \emph{resetting the clocks} in $Y$,
    \ie{}:
    $\forall x\in Y: \val_\ClockSet[Y:=0](x)=0$ and
    $\forall x\in\ClockSet\setminus Y:
        \val_\ClockSet[Y:=0](x)=\val_\ClockSet(x)$
, and  
    $\val[Y:=0]= (\val_\ClockSet[Y:=0], \val_\ParamSet)$.

The semantics of a Parametric Timed Game is defined
as a timed transition system on states, with timed and discrete transitions, and a winning condition.
A \emph{state} of a PTG is a pair $(\loc,\val)$ where $\loc$
is a location and $\val$ a valuation of the game satisfying its invariant:
$\val \models\Inv(\loc)$. The \emph{state space} is then
$\StateSpace = \{ (\loc,\val)\in \LocSet \times \ValSet \; | \;
\val \models \Inv(\loc) \}  = {\bigcup}_{\loc \in \LocSet}
\; \{\loc\} \times \Inv(\loc)$.
Let $\delay \in \mathbb{R}_{\geq 0}$ be a time delay. A \emph{timed transition} is
a relation ${\temptrans} \in \StateSpace \times \StateSpace$ s.t.
$\forall (\loc,\val),(\loc',\val')\in \StateSpace:
(\loc,\val)\temptrans (\loc',\val')$ iff
$\loc = \loc'$ and $\val' = \val + \delay$.
Let $\trans = (\loc,\guard,a,Y,\loc') \in \TransSet$
be a transition. A \emph{discrete transition} is a relation
${\disctrans} \in \StateSpace \times \StateSpace$ s.t.
$\forall (\loc,\val), (\loc',\val')\in \StateSpace:
(\loc,\val) \disctrans (\loc',\val')$ iff
$\val \models \guard$ and $ \val' =  \val[Y:=0]$.

States are grouped in symbolic states, similar to valuations grouped in zones.

\begin{definition}[symbolic state]\label{def:symbstate}
    A \emph{ symbolic state} of a PTG is a pair $\SymbState=(\loc,\zone)$ where $\loc$
    is a location and $\zone$ a zone of the game satisfying its invariant:
    $\zone \models \Inv(\loc)$.
    The symbolic state $(\loc,\zone)$ represents the subset of states $
    \{\loc\} \times \zone \subseteq \StateSpace$.

\end{definition}

\begin{definition}[temporal zone operators]
  For a zone $\zone$, its \emph{time-forward closure} and
  \emph{time-backward closure} are:
  \begin{align*}
    \TempSucc{\zone} &:= \{\, v + \delay \mid v \in \zone,\, \delay \geq 0 \,\}, \\
    \TempPred{\zone} &:= \{\, v \mid \exists\,\delay \geq 0 : v + \delay \in \zone \,\}.
  \end{align*}
\end{definition}

\begin{definition}[discrete zone operators]
  Let $\trans = (\loc, \guard, a, Y, \loc') \in \TransSet$.
  For a zone $\zone$:
  \begin{align*}
    \Succ(\trans, \zone) &:= \{\, v[Y:=0] \mid v \in \zone,\, v \models \guard \,\}, \\
    \Pred(\trans, \zone)     &:= \{\, v \mid v \models \guard,\; v[Y:=0] \in \zone \,\}.
  \end{align*}
\end{definition}

\begin{definition}[safe temporal predecessor]
  For zones $W$ and $U$:
  \[
    \SafePred(W, U) :=
      \bigl\{\, v \;\bigm|\;
        \exists\,\delay \geq 0 :\;
        v + \delay \in W \;\land\;
        \forall\,\delay' \in [0,\delay] : v + \delay' \notin U
      \,\bigr\}.
  \]
  That is, $v \in \SafePred(W, U)$ if time can elapse from $v$ to reach $W$
  without passing through $U$.
\end{definition}

\begin{remark}[zone operators on symbolic states]\label{rem:zone-ops-lift}
  The operators above are defined on zones (sets of valuations).
  They lift to symbolic states by intersecting with the location invariant
  (so that the zone of the result satisfies the invariant, as required by \cref{def:symbstate}):
  for $\SymbState = (\loc, \zone)$ and $\trans = (\loc, \guard, a, Y, \loc')$:
  \begin{align*}
    \TempSucc{\SymbState} &:= (\loc,\; \TempSucc{\zone} \cap \Inv(\loc)), \\
    \TempPred{\SymbState} &:= (\loc,\; \TempPred{\zone} \cap \Inv(\loc)), \\
    \Succ(\trans, \SymbState) &:= (\loc',\; \Succ(\trans, \zone) \cap \Inv(\loc')).
  \end{align*}
  $\Pred$ and $\SafePred$ are used directly on zones throughout.
  We will not distinguish a zone from its corresponding set of valuations in our notation.
\end{remark}

Let $\vec{0}$ be the clock valuation where all clocks have value $0$.
The set of possible initial states of the PTG is
$\SymbState_0 = \{ (\loc_0, (\vec{0}, \val_\ParamSet)) \; | \;
\val_\ParamSet \in \ParamValSet{\ParamSet}:
(\vec{0}, \val_\ParamSet) \models \Inv(\loc_0) \}$.

Algorithms for PTA (and thus PTG) explore and operate on a Parametric Zone Graph (PZG).
\begin{definition}[parametric zone graph]
    A \emph{Parametric Zone Graph} (PZG) for a PTG
    $G=(\LocSet , \ClockSet, \ParamSet, \LabelSet, \TransSet_c,
    \TransSet_u, \loc_0, \Inv)$
    is a tuple $(\SymbStateSet, \SymbStateInit,
    \Rightarrow^\trans_c,\Rightarrow^\trans_u)$
    where $\SymbStateSet \subseteq 2^\StateSpace$ is a set of symbolic states,
    $\SymbStateInit \in \SymbStateSet$ is the initial symbolic state,
    and ${\Rightarrow^\trans_c}, {\Rightarrow^\trans_u} \subseteq \SymbStateSet \times \SymbStateSet$ are transition relations labeled by transitions of~$G$.
\end{definition}

\begin{definition}[canonical PZG]\label{def:canonical-pzg}
    The \emph{canonical PZG} of a PTG~$G$ is the PZG $(\SymbStateSet, \SymbStateInit,
    \Rightarrow^\trans_c,\Rightarrow^\trans_u)$ where $\SymbStateInit = \TempSucc{\SymbState_0}$ and the transition relations are defined by:
    $\SymbState \Rightarrow^t_c \SymbState'$ iff $\SymbState' = \TempSucc{\Succ(\trans,\SymbState)}$ and $\trans\in T_c$; and
    $\SymbState \Rightarrow^t_u \SymbState'$ iff $\SymbState' = \TempSucc{\Succ(\trans,\SymbState)}$ and $\trans\in T_u$;
    and $\SymbStateSet$ is the smallest set containing $\SymbStateInit$ and closed under these transitions.
\end{definition}
The canonical PZG exists and is unique: $\SymbStateSet$ is well-defined as the least fixed point of forward exploration from $\SymbStateInit$, and the transitions are deterministic (each pair $(\SymbState, \trans)$ yields exactly one successor), so the structure is uniquely determined by~$G$.

\begin{definition}[zones at a location]
For a location $\loc$ (with $\ClockSet$, $\ParamSet$, and $\SymbStateSet$ implicit from the PZG under discussion), the associated zone set is
\[
\ZonesAt{\loc} \;:=\;
\{\, \zone \in \ZoneSet(\ClockSet, \ParamSet) \mid (\loc,\zone) \in \SymbStateSet \,\}.
\]
\end{definition}

We study PTGs with a \emph{reachability objective} given by a target set $\TargetSet \subseteq \LocSet$.
The controller is considered winning from a state if there exists a memoryless strategy for the controller such that, regardless of how the environment plays, the game eventually reaches a location in $\TargetSet$.
Solving a PTG then means characterizing the parameter valuations for which the controller is winning from the initial state $\SymbState_0$. A detailed definition of winning strategies can be found in~\cite{DBLP:conf/tacas/DahlsenJensenFPP24}.

%% file: abstractions.tex

\label{sec:statespacered}

\subsection{General abstraction framework}
We now formally define a general framework for state space abstractions in PTG and prove its correctness in \cref{sec:correctness}. The remaining subsections instantiate the framework with concrete abstraction functions, recovering known techniques and introducing new ones.

\begin{definition}[abstraction function]
Given a PZG with symbolic state set $\SymbStateSet$, an \emph{abstraction function} is a function
\( \mergefun : \LocSet \to 2^{\ZoneSet(\ClockSet,\ParamSet)} \).
For every location $\loc$, it must over-approximate the original zones:
\[
\forall \zone \in \ZonesAt{\loc}\;
\exists \zone' \in \mergefun(\loc)
    \text{ such that } \zone \subseteq \zone'.
\]
Equivalently, every zone in $\ZonesAt{\loc}$ must be covered by some zone in $\mergefun(\loc)$.
\end{definition}
The simplest instance, which serves as a baseline, is the trivial abstraction function.
\begin{example}
    $\ftriv$ is the abstraction function that leaves the zone set unchanged:
    $$\ftriv(\loc) = \ZonesAt{\loc} $$
\end{example}

To define the reduced PZG, we need to know which abstract zone each original zone maps to. This is captured by the representative map induced by $\mergefun$.
\begin{definition}[representative map]
    Given an abstraction function $\mergefun$, let $rep_{\mergefun}(\loc, \zone)$ be an arbitrary but fixed zone $\zone'\in \mergefun(\loc)$ such that $\zone \subseteq \zone'$ (non-empty by coverage).
    Define the \emph{representative map} as the function taking symbolic states and returning their representative in the reduced symbolic state space:
    $$\repmap{\mergefun}(\loc,\zone) := (\loc, rep_{\mergefun}(\loc,\zone))$$
\end{definition}

\begin{definition}[reduced parametric zone graph]
Let $(\SymbStateSet, \SymbStateInit, \Rightarrow_c^t, \Rightarrow_u^t)$ be a PZG and $\mergefun : \LocSet \to 2^{\ZoneSet(\ClockSet,\ParamSet)}$ an abstraction function with representative map $\repmap{\mergefun}$.
The reduced transition relation for controllable transitions is
$$\Rightarrow_{c,red}^t := \{(\repmap{\mergefun}(\SymbState), \repmap{\mergefun}(\SymbState')) \; \mid \; \SymbState \Rightarrow_c^t \SymbState'\}$$
and analogously for uncontrollable transitions:
$$\Rightarrow_{u,red}^t := \{(\repmap{\mergefun}(\SymbState), \repmap{\mergefun}(\SymbState')) \; \mid \; \SymbState \Rightarrow_u^t \SymbState'\}$$
Together with the abstract symbolic states $\{\repmap{\mergefun}(\SymbState) \mid \SymbState \in \SymbStateSet\}$, these define the \emph{reduced parametric zone graph}, itself a PZG for~$G$:
$$\PZGred{\mergefun} = (\{\repmap{\mergefun}(\SymbState) \mid \SymbState \in \SymbStateSet\},\; \repmap{\mergefun}(\SymbStateInit),\;\Rightarrow_{c,red}^t,\;\Rightarrow_{u,red}^t)$$
\end{definition}

\subsection{PTG Reachability Algorithm}\label{sec:algorithm}
In~\cite{DBLP:conf/wodes/JovanovicFLR12,DBLP:journals/ijcon/JovanovicLR19} a semi-algorithm was introduced and later implemented in~\cite{DBLP:conf/tacas/DahlsenJensenFPP24} that synthesizes the set of parameters for which a winning strategy exists. We recall this algorithm here (\cref{alg:orig}), as it forms the basis for the correctness proofs that follow. For the purposes of this paper, the algorithm is presented with the PZG as an explicit input, rather than deriving the canonical PZG from the PTG internally; this makes the abstraction substitution in \cref{sec:correctness} precise.

We briefly explain the algorithm. The idea is to compute reachable states while propagating winning conditions backward over controllable and uncontrollable transitions. The main loop (\cref{alg:line:mainloop,alg:line:functionCall}) calls \textsc{explore} or \textsc{update} until no states remain. Procedure \textsc{explore} picks a state $\SymbState{}$ from $\WaitingExplore{}$ (\cref{alg:line:explore:pick}) and iterates over its PZG transitions (\crefrange{alg:line:explore:forbegin}{alg:line:explore:forend}): each successor $\SymbState{}'$ records $\SymbState{}$ in $\Depends$ and, if new, is queued for exploration. A target state is marked winning, queuing its predecessors for update (\crefrange{alg:line:explore:beginif}{alg:line:explore:endif}); in any case $\SymbState{}$ is queued for update (\cref{alg:line:explore:toupdate}) and marked explored (\cref{alg:line:explore:explored}).

Procedure \textsc{update} back-propagates winning zones for $\SymbState{}$ from $\WaitingUpdate{}$. It collects in $\Uncontrollable{}$ the predecessors (\Pred) of the non-winning part of the uncontrollable successors (\cref{alg:line:update:uncontrol}), and in $\WinningMove{}$ the predecessors of the winning part of the controllable successors (\cref{alg:line:update:winmove}). The safe temporal predecessors (\SafePred) of $\Win[\SymbState] \cup \WinningMove{}$ that can let time elapse into $\SymbState{}$ while avoiding $\Uncontrollable{}$ (\cref{alg:line:update:safepred}) are winning. On a change to $\Win[\SymbState]$, the predecessors of $\SymbState{}$ are re-queued (\crefrange{alg:line:update:startend}{alg:line:update:finishend}), and the parameter projection at $\SymbState{}_I$ is reported (\cref{line:report}).

\input{orig_algo}

\noindent\textbf{Role of the abstraction.}
The abstraction changes only which symbolic states are explored and how they are grouped, not the game's moves: \textsc{update} applies $\Pred$ and $\SafePred$ to the original guards and resets, so the reduced PZG is only a skeleton carrying the \emph{exact} backward fixed point.
Thus abstract zones never enlarge the winning region (recovered exactly on each original zone, \cref{lem:zone-corr}), and, unlike game abstractions~\cite{DBLP:conf/sas/HenzingerMMR00,DBLP:conf/lics/AlfaroGJ04}, no asymmetric handling of the players is needed: both predecessors are concrete, so coverage alone gives soundness and completeness.

\subsection{Correctness of state space abstractions}\label{sec:correctness}
\input{correctness}

\subsection{Instantiating the framework}\label{sec:instantiations}
The remaining subsections instantiate the framework with concrete abstraction methods.
Each method is specified as a type of \emph{abstraction step}: a local operation that produces a valid abstraction function from an existing one.
This formulation avoids requiring the full zone set $\ZonesAt{\loc}$ (which may be infinite) and matches the incremental nature of on-the-fly exploration.

\begin{definition}[abstraction step]\label{def:abstraction-step}
Let $\mergefun$ be an abstraction function. An \emph{abstraction step} at location $\loc$ produces a new function $\mergefun' : \LocSet \to 2^{\ZoneSet(\ClockSet,\ParamSet)}$ satisfying:
\begin{enumerate}
    \item $\mergefun'(\loc') = \mergefun(\loc')$ for all $\loc' \neq \loc$ \emph{(locality)},
    \item for every $\zone \in \mergefun(\loc)$, there exists $\zone' \in \mergefun'(\loc)$ with $\zone \subseteq \zone'$ \emph{(coarsening)}.
\end{enumerate}
An \emph{abstraction sequence} is a sequence $(\mergefun_i)_{i \geq 0}$ where $\mergefun_0 = \ftriv$ and each $\mergefun_{i+1}$ is obtained from $\mergefun_i$ by an abstraction step.
\end{definition}

\begin{lemma}\label{lem:abstraction-sequence}
Every element of an abstraction sequence is a valid abstraction function.
\end{lemma}
\begin{proof}
$\ftriv$ satisfies coverage. Each step preserves coverage: if $\zone \in \ZonesAt{\loc}$ was covered by $\zone' \in \mergefun(\loc)$, then coarsening gives $\zone'' \in \mergefun'(\loc)$ with $\zone \subseteq \zone' \subseteq \zone''$.
\end{proof}

\begin{remark}
For finite PZGs, each abstraction sequence stabilizes after finitely many steps;
the closed-form characterizations given alongside each step type below describe these fixed points.
In general, the result of applying steps may depend on the order in which they are applied, but every abstraction function produced by a step sequence is sound by \cref{thm:correctness}.
\end{remark}

\subsection{Zone inclusion check}
The idea is to prune redundant symbolic states: if two zones $\zone \subseteq \zone'$
appear at the same location, $\zone$ is already covered by $\zone'$ and can be discarded.
We consider two variants depending on whether this check is applied in one or both directions.

\begin{definition}[inclusion step]
Let $\mergefun$ be an abstraction function and $\zone \in \mergefun(\loc)$.
If there exists $\zone' \in \mergefun(\loc)$ with $\zone \neq \zone'$ and $\zone \subseteq \zone'$,
then removing $\zone$ is an \emph{inclusion step}:
$$\mergefun'(\loc) = \mergefun(\loc) \setminus \{\zone\}.$$
\end{definition}
Coarsening is immediate: $\zone$ itself is covered by $\zone'$, and all other zones are unchanged.

\begin{remark}
An inclusion step can be applied in two directions during on-the-fly exploration.
\emph{Forward}: when a newly discovered zone $\zone$ is already subsumed by an existing $\zone' \in \mergefun(\loc)$, discard $\zone$ immediately.
\emph{Backward}: when a newly discovered zone $\zone'$ subsumes an existing $\zone \in \mergefun(\loc)$, retroactively remove $\zone$.
Both are the same formal step; they differ only in which zone is new and which is old.
\end{remark}

When $\ZonesAt{\loc}$ is finite, a maximal sequence of inclusion steps applied in \emph{both} directions (\emph{double inclusion}) stabilizes to the \emph{antichain} of $\subseteq$-maximal zones:
$$\finc(\loc) = \{ \zone \in \ZonesAt{\loc} \mid \nexists\, \zone' \in \ZonesAt{\loc} : \zone \subsetneq \zone'\}.$$
Applying steps in only one direction may stabilize earlier: forward-only never removes an existing zone, and backward-only never discards a new zone subsumed by an already-present one.

\subsection{Zone merging}
Zone merging, introduced for parametric timed automata in~\cite{DBLP:conf/atva/AndreFS13,DBLP:conf/formats/AndreMPP22},
reduces the state space by replacing two zones at the same location with their union,
when that union is already convex (so no new states are introduced).

\begin{definition}[merge step]
Let $\mergefun$ be an abstraction function and $\zone_1, \zone_2 \in \mergefun(\loc)$ where $\zone_m = \zone_1 \cup \zone_2$ is convex~\cite{DBLP:conf/atva/AndreFS13}.
A \emph{merge step} replaces them with $\zone_m$:
$$\mergefun'(\loc) = (\mergefun(\loc) \setminus \{\zone_1, \zone_2\}) \cup \{\zone_m\}.$$
\end{definition}
Coarsening holds: $\zone_1, \zone_2 \subseteq \zone_m = \zone_1 \cup \zone_2$, and all other zones are unchanged.

\subsection{Location-based}
The location-based abstraction is the coarsest instantiation of the framework:
every zone at location $\loc$ is replaced by the invariant $\Inv(\loc)$, yielding
a single abstract state per location.
\begin{definition}[\floc]
$$\floc(\loc) = \{\Inv(\loc)\}$$
\end{definition}
Coverage holds trivially: every symbolic state $(\loc, \zone)$ satisfies
$\zone \subseteq \Inv(\loc)$ by definition, so the single zone $\Inv(\loc)$ covers
all zones at $\loc$.
Unlike the other instantiations, $\floc$ does not require the abstraction step
machinery: the characterization above is well-defined for any PZG,
including infinite ones, since it depends only on the invariant of each location,
not on enumerating the reachable zones.
A similar location-based abstraction is used in~\cite{DBLP:conf/concur/JensenLLS25}
for timed ATL.

\subsection{Hull abstractions}
Unlike zone merging, which only combines zones when their union is already convex,
hull abstractions always collapse zones into a single over-approximating zone, potentially introducing new valuations.
The variants differ in the class of zones they produce, trading precision for computational cost.

\begin{definition}[hull operator]
A \emph{hull operator} is a function $\mathcal{H} : 2^{\ZoneSet(\ClockSet,\ParamSet)} \to
\ZoneSet(\ClockSet,\ParamSet)$ satisfying $\bigcup_{Z \in S} Z \subseteq \mathcal{H}(S)$ for
every finite set of zones $S$.
\end{definition}

\begin{definition}[hull step]
Let $\mergefun$ be an abstraction function, $\zone \in \mergefun(\loc)$ a newly discovered zone, and $H \in \mergefun(\loc)$ an existing zone at $\loc$.
A \emph{hull step} replaces both with the hull of the pair:
$$\mergefun'(\loc) = (\mergefun(\loc) \setminus \{H, \zone\}) \cup \{\mathcal{H}(\{H, \zone\})\}.$$
\end{definition}
Coarsening holds: $H, \zone \subseteq \mathcal{H}(\{H, \zone\})$, and all other zones are unchanged.
Successive hull steps at a location accumulate into a running hull: after processing zones $\zone_1, \ldots, \zone_k$, the abstract zone at $\loc$ is $\mathcal{H}(\{\ldots\mathcal{H}(\{\zone_1, \zone_2\}), \ldots, \zone_k\})$.
The result is order-independent for associative hull operators (convex hull, octagonal hull, box hull).

\subsubsection{Fixed-point characterization (finite PZGs)}
When $\ZonesAt{\loc}$ is finite, any maximal sequence of hull steps at $\loc$ stabilizes to a single zone.
For a hull operator $\mathcal{H}$, the fixed point is:
$$\fhullgen(\loc) := \bigl\{\, \mathcal{H}(\ZonesAt{\loc}) \,\bigr\}.$$

\subsubsection{Instances}
The three principal instances of $\mathcal{H}$, ordered from finest to coarsest, are:
\begin{description}
  \item[$\fhull$ (convex hull)]
    $\mathcal{H}_{\mathrm{conv}}(S)$ is the smallest convex zone containing $\bigcup S$.

  \item[$\focthull$ (octagonal hull)]
    $\mathcal{H}_{\mathrm{oct}}(S)$ is the smallest octagonal zone containing $\bigcup S$.
    Octagonal zones extend difference-bound constraints with terms of the form
    $x_i \pm x_j \leq c$, capturing more diagonal structures than a bounding box.

  \item[$\fboxhull$ (box hull)]
    $\mathcal{H}_{\mathrm{box}}(S)$ is the smallest axis-aligned bounding zone containing
    $\bigcup S$, using only individual-variable bounds $x_i \leq c$ and $x_i \geq c$.
    It is the cheapest to compute but the coarsest.
\end{description}

These satisfy $\mathcal{H}_{\mathrm{conv}}(S) \subseteq \mathcal{H}_{\mathrm{oct}}(S) \subseteq \mathcal{H}_{\mathrm{box}}(S)$
for any $S$, so $\fhull$ is the finest and $\fboxhull$ the coarsest hull abstraction.

%% file: orig_algo.tex

\begin{algorithm}
  \caption{For PTG $\game = (\LocSet, \ClockSet, \ParamSet, \LabelSet, \TransSet_c,
  \TransSet_u, \loc_0, \Inv)$ and PZG $(\SymbStateSet, \SymbStateInit, \Rightarrow_c^t, \Rightarrow_u^t)$, and reachability objective $\TargetSet$,
  returns the set of all parameter valuations that win the game.\label{alg:orig}}
  \begin{algorithmic}[1]
    \State $\Explored, \WaitingUpdate, \WaitingExplore \gets \emptyset, \emptyset, \{\SymbStateInit\}$
    \Comment{Symbolic state sets}
    \State $\mathit{Win} := \{\}$
    \Comment {Map from symbolic states to unions of zones}
    \State $\Depends := \{\}$
    \Comment {Map from symbolic states to sets of symbolic states}
    \State $\WinningParam := False$
    \Statex
    \Function{solvePTG}{}
    \While{$\WaitingExplore \neq \emptyset \lor \WaitingUpdate \neq \emptyset$}
    \label{alg:line:mainloop}
    \State Choose either \Call{explore}{ \hspace*{-1mm}} or  \Call{update}{ \hspace*{-1mm}}
    \label{alg:line:functionCall}\EndWhile
    \State \Return $\WinningParam$
    \EndFunction
    \Statex

    \Procedure{explore}{}
    \State $\SymbState \gets extract(\WaitingExplore)$ \label{alg:line:explore:pick}
    \For{$\SymbState \Rightarrow^{\trans} \SymbState'$ in PZG :} \label{alg:line:explore:forbegin}
    \State $\Depends[\SymbState'] \gets \Depends[\SymbState'] \cup \{\SymbState\}$
    \If{$\SymbState'$ not in $\Explored$}
    \State $\WaitingExplore \gets \WaitingExplore \cup \{\SymbState'\}$
    \EndIf
    \EndFor \label{alg:line:explore:forend}
    \If{$\SymbState.\loc \in \TargetSet$} \label{alg:line:explore:beginif}
    \State $\mathit{Win}[\SymbState] \gets \SymbState.\zone$
    \State $\WaitingUpdate \gets \WaitingUpdate \cup \Depends[\SymbState]$\label{alg:line:explore:winupdate}
    \EndIf\label{alg:line:explore:endif}
    \State $\WaitingUpdate \gets \WaitingUpdate \cup \{\SymbState\}$ \label{alg:line:explore:toupdate}
    \State $\Explored \gets \Explored \cup \{\SymbState\}$  \label{alg:line:explore:explored}
    \EndProcedure
    \Statex
    \Procedure{update}{}
    \State $\SymbState \gets extract(\WaitingUpdate)$
    \State $\Uncontrollable \gets \hspace{-1em}\underset{\{(\SymbState', t) \mid  \SymbState
    \Rightarrow^{t}_u\SymbState'\}}{\bigcup}\hspace{-1em} Pred (t, \SymbState' \setminus
    \mathit{Win}[\SymbState'])$\label{alg:line:update:uncontrol}
    \State $\WinningMove \gets \hspace{-1em}\underset{\{(\SymbState', t) \mid  \SymbState
    \Rightarrow^{t}_c\SymbState'\}}{\bigcup}\hspace{-1em} Pred (t, \mathit{Win}[\SymbState'])$\label{alg:line:update:winmove}
    \State $\NewWin := \SafePred( Win[\SymbState] \cup \WinningMove, \: \Uncontrollable) \:  \cap \: \SymbState$\label{alg:line:update:safepred}
    \If{$\NewWin \not\subseteq \mathit{Win}[\SymbState]$}\label{alg:line:update:startend}
    \State $\WaitingUpdate \gets \WaitingUpdate \cup \Depends[\SymbState]$
    \State $\mathit{Win}[\SymbState] \gets  \mathit{Win}[\SymbState] \cup \NewWin$
    \EndIf\label{alg:line:update:finishend}
    \State $\WinningParam \gets \ProjectParam{(\mathit{Win}[\SymbState] \cap \SymbState_0)}$ \label{line:report}
    \EndProcedure
  \end{algorithmic}
\end{algorithm}

%% file: correctness.tex
The goal of this section is to show that the algorithm (\cref{alg:orig}) produces the same winning parameter set whether it receives the canonical PZG or the reduced PZG $\PZGred{\mergefun}$ as input.
The argument rests on two properties.
First, every original zone is contained in an abstract zone (coverage requirement).
Second, PZG transitions are \emph{image-preserving}: for any symbolic state $(\loc,\zone)$ and transition $\trans$ from it with successor zone $\zone_{\mathrm{img}}$, every $v \in \zone$
where $\trans$ is enabled satisfies $v[Y:=0] \in \zone_{\mathrm{img}}$,
since $\zone_{\mathrm{img}}$ is by definition the set of images of $\trans$
applied to $\zone$.
In particular, states in $\zone$ always land in $\zone_{\mathrm{img}}$, not
in the possibly larger abstract successor
$\zone'_{\mathrm{img}} \supseteq \zone_{\mathrm{img}}$.

Throughout this section, whenever a transition
$\trans = (\loc, g, a, Y, \loc')$ appears, we write $g$, $Y$, and
$\loc'$ for its guard, reset clocks, and target location; subscripted variants
($g_u, Y_u, \loc'_u$ for $\trans_u$, etc.)\ follow the same convention.

We write $\Win^*[(\loc,\zone)]$ for the set of valuations eventually marked
winning at symbolic state $(\loc,\zone)$ in the canonical PZG (under fair
exploration), and $\Win^*_\alpha[(\loc,\zone')]$ for the corresponding set
computed on $\PZGred{\mergefun}$.

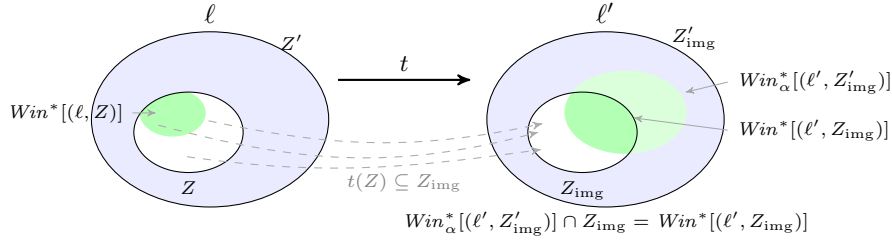
\begin{figure}[t]
\centering
\begin{tikzpicture}[>=stealth', font=\small, scale=0.95]
  \colorlet{abstractfill}{blue!8}
  \colorlet{winfill}{green!30}
  \colorlet{winfilldim}{green!14}

  \node[above, font=\normalsize] at (0, 1.22) {$\loc$};

  \fill[abstractfill] (0,0) ellipse (1.65 and 1.22);
  \draw (0,0) ellipse (1.65 and 1.22);
  \node[anchor=south west, font=\scriptsize] at (0.85, 0.85) {$\zone'$};

  \fill[white] (-0.3,-0.18) ellipse (0.76 and 0.56);

  \begin{scope}
    \clip (-0.3,-0.18) ellipse (0.76 and 0.56);
    \fill[winfill] (-0.52, 0.09) ellipse (0.46 and 0.33);
  \end{scope}

  \draw (-0.3,-0.18) ellipse (0.76 and 0.56);
  \node[anchor=north, font=\scriptsize] at (-0.3, -0.74) {$\zone$};

  \node[anchor=east, font=\scriptsize] at (-1.08, 0.09) {$\Win^*[(\loc,\zone)]$};
  \draw[->, very thin, gray!70] (-1.08, 0.09) -- (-0.76, 0.09);

  \node[above, font=\normalsize] at (5.5, 1.22) {$\loc'$};

  \fill[abstractfill] (5.5, 0) ellipse (1.65 and 1.22);
  \draw (5.5, 0) ellipse (1.65 and 1.22);
  \node[anchor=south west, font=\scriptsize] at (6.3, 0.85) {$\zone'_{\mathrm{img}}$};

  \begin{scope}
    \clip (5.5, 0) ellipse (1.65 and 1.22);
    \fill[winfilldim] (5.78, 0.10) ellipse (0.85 and 0.58);
  \end{scope}

  \fill[white] (5.18,-0.18) ellipse (0.76 and 0.56);

  \begin{scope}
    \clip (5.18,-0.18) ellipse (0.76 and 0.56);
    \fill[winfill] (5.78, 0.10) ellipse (0.85 and 0.58);
  \end{scope}

  \draw (5.18,-0.18) ellipse (0.76 and 0.56);
  \node[anchor=north, font=\scriptsize] at (5.18, -0.74) {$\zone_{\mathrm{img}}$};

  \node[anchor=west, font=\scriptsize] at (7.2, 0.52)
    {$\Win^*_\alpha[(\loc',\zone'_{\mathrm{img}})]$};
  \draw[->, very thin, gray!70] (7.2, 0.52) -- (6.62, 0.38);

  \node[anchor=west, font=\scriptsize] at (7.2, -0.15)
    {$\Win^*[(\loc',\zone_{\mathrm{img}})]$};
  \draw[->, very thin, gray!70] (7.2, -0.15) -- (5.90, 0.02);

  \node[font=\scriptsize, align=center] at (5.5, -1.45)
    {$\Win^*_\alpha[(\loc',\zone'_{\mathrm{img}})] \cap \zone_{\mathrm{img}}
      = \Win^*[(\loc',\zone_{\mathrm{img}})]$};

  \draw[->, thick] (1.78, 0.55) -- (3.62, 0.55) node[midway, above] {$t$};

  \draw[->, dashed, gray!65]
    (-0.02, -0.02) to[out=-12, in=192] (4.62, -0.06);
  \draw[->, dashed, gray!65]
    (-0.32, -0.52) to[out=-8,  in=188] (4.60, -0.42);
  \draw[->, dashed, gray!65]
    (-0.72, -0.08) to[out=-15, in=195] (4.56, -0.18);

  \node[font=\scriptsize, gray, anchor=north] at (2.7, -0.62)
    {$t(\zone) \subseteq \zone_{\mathrm{img}}$};

\end{tikzpicture}
\caption{Consequence of Lemma~\ref{lem:zone-corr}.
  $\Win^*_\alpha[(\loc',\zone'_{\mathrm{img}})]$ (light and dark green) extends
  into the region $\zone'_{\mathrm{img}} \setminus \zone_{\mathrm{img}}$,
  but states from $\zone$ map only into $\zone_{\mathrm{img}}$ (image-preserving,
  dashed arrows).
  Restricted to $\zone_{\mathrm{img}}$, the abstract and full winning sets
  coincide (dark green), so the predecessor computation restricted to $\zone$
  recovers exactly $\Win^*[(\loc,\zone)]$.}
\label{fig:zone-corr}
\end{figure}

\begin{lemma}[Zone correspondence]\label{lem:zone-corr}
  Let $\mergefun$ be an abstraction function.
  For every abstract symbolic state $(\loc, \zone') \in \PZGred{\mergefun}$ and
  every original symbolic state $(\loc, \zone)$ in the canonical PZG with
  $\repmap{\mergefun}(\loc,\zone) = (\loc,\zone')$:
  \begin{equation}\label{eq:zone-corr}
    \Win^*_\alpha[(\loc, \zone')] \;\cap\; \zone \;=\; \Win^*[(\loc, \zone)].
  \end{equation}
\end{lemma}

\begin{proof}
  We show both inclusions separately, using the correctness theory
  of~\cite{DBLP:conf/tacas/DahlsenJensenFPP24}.
  For each location $\loc$, write $W(\loc)$ for the set of valuations
  from which the controller has a winning strategy at~$\loc$.
  By~\cite{DBLP:conf/tacas/DahlsenJensenFPP24},
  the PTG algorithm is sound at every step (Invariant~4: every state
  marked winning is genuinely winning) and complete in the limit
  (Theorem~3: under fair exploration, every winning state is eventually
  discovered).
  In particular, $\Win^*[(\loc,\zone)] = W(\loc) \cap \zone$ for every
  symbolic state $(\loc,\zone)$ in the canonical PZG.

  \medskip\noindent\textbf{Soundness}
  ($\Win^*_\alpha[(\loc,\zone')] \cap \zone \subseteq \Win^*[(\loc,\zone)]$).
  We show that every valuation added to $\Win_\alpha[(\loc,\zone')]$ is
  genuinely winning at~$\loc$, by induction on the order in which
  valuations are added.
  At target locations ($\loc \in \TargetSet$), the algorithm initializes
  $\Win_\alpha[(\loc,\zone')] \gets \zone'$; every $v \in \zone'$ at a
  target location is winning by definition.
  At non-target locations, a valuation $v$ is added only through the
  $\SafePred$ computation, which requires:
  (i)~a controllable transition $\trans$ and a delay $\delta \geq 0$ such that
  $v + \delta \models g$, with $(v+\delta)[Y\!:=\!0]$ in a set already marked winning;
  and (ii)~every uncontrollable successor during $[0,\delta]$ is also in a
  set already marked winning.
  Since abstract transitions inherit their guards and resets from genuine
  PZG transitions, conditions (i) and (ii) correspond to real game moves.
  By the inductive hypothesis, the successor sets contain only genuinely
  winning valuations, so $v$ is genuinely winning.
  Therefore, $\Win^*_\alpha[(\loc,\zone')] \subseteq W(\loc)$.
  Intersecting with $\zone \subseteq \zone'$:
  $\Win^*_\alpha[(\loc,\zone')] \cap \zone \subseteq W(\loc) \cap \zone
  = \Win^*[(\loc,\zone)]$.

  \medskip\noindent\textbf{Completeness}
  ($\Win^*[(\loc,\zone)] \subseteq \Win^*_\alpha[(\loc,\zone')] \cap \zone$).
  Let $v \in \Win^*[(\loc,\zone)] = W(\loc) \cap \zone$.
  Since $\zone \subseteq \zone'$ (coverage), $v \in \zone'$, so $v$ is present
  in the abstract symbolic state $(\loc,\zone')$.
  We verify the completeness argument
  of~\cite{DBLP:conf/tacas/DahlsenJensenFPP24} (Theorem~3)
  for $\PZGred{\mergefun}$, by induction on the discrete distance
  $n(\loc, v)$ to the target: the minimum, over winning strategies
  from $(\loc, v)$, of the maximum number of discrete transitions
  along any play consistent with that strategy.

  \emph{Base case ($n = 0$):}
  $v$ is at a target location $\loc \in \TargetSet$.
  The algorithm sets $\Win_\alpha[(\loc,\zone')] \gets \zone'$ upon exploration,
  so $v \in \zone' \subseteq \Win_\alpha[(\loc,\zone')]$.

  \emph{Inductive step ($n > 0$):}
  The winning strategy from $v$ waits time $\delta \geq 0$ and then takes a
  controllable transition $\trans$ to location $\loc'$, reaching state
  $v' = (v+\delta)[Y\!:=\!0]$ with $n(\loc', v') < n$.
  We show in two parts that $v$ is eventually added to
  $\Win_\alpha[(\loc,\zone')]$.

  \emph{(a) Lifting successors to the abstract graph.}
  Since $v \in \zone \subseteq \zone'$, every original PZG transition
  from $(\loc,\zone)$ lifts to an abstract transition from $(\loc,\zone')$
  whose successor zone contains the original.
  For the controllable transition $\trans$:
  $(\loc,\zone) \Rightarrow_c^\trans (\loc', \zone_{\mathrm{img}})$
  gives $(\loc,\zone') \Rightarrow_{c,red}^\trans (\loc', \zone'_{\mathrm{img}})$
  with $\zone'_{\mathrm{img}} = \repmap{\mergefun}(\loc', \zone_{\mathrm{img}})
  \supseteq \zone_{\mathrm{img}}$.
  Since $v + \delta \in \zone$ (PZG zones include time successors) and
  $v + \delta \models g$, we have
  $v' = (v+\delta)[Y\!:=\!0] \in \zone_{\mathrm{img}} \subseteq
  \zone'_{\mathrm{img}}$ (image-preserving).
  As $n(\loc', v') < n$, the inductive hypothesis gives that $v'$ is eventually
  in $\Win_\alpha[(\loc', \zone'_{\mathrm{img}})]$.
  For each uncontrollable transition $\trans_u$ enabled from
  $v + \delta'$ ($\delta' \in [0,\delta]$), the same lifting applies:
  the concrete successor $v'' = (v+\delta')[Y_u\!:=\!0]$ reaches
  location $\loc'_u$ and lies in the original successor zone
  $\zone_{\mathrm{img}}^u$ (image-preserving), which lifts to an abstract successor
  $\zone'^u_{\mathrm{img}} \supseteq \zone_{\mathrm{img}}^u$.
  Since $v$ is genuinely winning, $v''$ has distance
  $n(\loc'_u, v'') \leq n - 1$ (the uncontrollable move consumes one
  discrete transition), so by the inductive hypothesis $v''$ is eventually
  in $\Win_\alpha$ at the corresponding abstract successor.

  \emph{(b) Backward propagation.}
  Once all successors are resolved, the algorithm's update step
  (\cref{alg:line:update:safepred}) at $(\loc,\zone')$ computes
  $\NewWin = \SafePred(\Win_\alpha[(\loc,\zone')] \cup \WinningMove,\;
  \Uncontrollable) \cap \zone'$, where:
  \begin{itemize}
  \item \emph{Controllable:}
    \[
    \WinningMove = \!\!\bigcup_{(\loc,\zone') \Rightarrow_{c,red}^{\trans_c}
    (\loc'_c,\zone'_c)}\!\! \Pred(\trans_c,\, \Win_\alpha[(\loc'_c,\zone'_c)]).
    \]
    Since $v' \in \Win_\alpha[(\loc', \zone'_{\mathrm{img}})]$ and
    $v + \delta \models g$ with $(v+\delta)[Y\!:=\!0] = v'$, we get
    $v + \delta \in \Pred(\trans,\, \Win_\alpha[(\loc',\zone'_{\mathrm{img}})])
    \subseteq \WinningMove$.

  \item \emph{Uncontrollable:}
    \[
    \Uncontrollable = \!\!\bigcup_{(\loc,\zone') \Rightarrow_{u,red}^{\trans_u}
    (\loc'_u,\zone'_u)}\!\! \Pred(\trans_u,\, \zone'_u \setminus
    \Win_\alpha[(\loc'_u,\zone'_u)]).
    \]
    Since every uncontrollable successor from the waiting interval
    $[0,\delta]$ is now in $\Win_\alpha$ at its abstract successor,
    $v + \delta' \notin \Uncontrollable$ for all $\delta' \in [0,\delta]$.
  \end{itemize}
  With both conditions met: $v + \delta \in \WinningMove$ and
  $v + \delta' \notin \Uncontrollable$ for $\delta' \in [0,\delta]$.
  By definition of $\SafePred$, this gives $v \in \NewWin$, so
  $v$ is added to $\Win_\alpha[(\loc,\zone')]$.
  Under fair exploration, this update happens in finite time.

  \medskip
  Combining both directions gives
  $\Win^*_\alpha[(\loc,\zone')] \cap \zone = \Win^*[(\loc,\zone)]$.
\end{proof}

\begin{remark}[Consequence]
  \Cref{lem:zone-corr} says that merging zones loses no precision on the states
  that were actually there: $\Win^*_\alpha[(\loc,\zone')]$ restricted to the
  original zone $\zone \subseteq \zone'$ recovers exactly
  $\Win^*[(\loc,\zone)]$ (see \cref{fig:zone-corr}).
\end{remark}

\begin{theorem}[Correctness of reduced PZG]\label{thm:correctness}
  Let $\mergefun$ be an abstraction function satisfying the coverage requirement.
  Under fair exploration, running the PTG algorithm on $\PZGred{\mergefun}$
  converges to the same winning parameter set as running it on the canonical PZG.
  In particular, if either algorithm terminates, both yield the same result;
  if the PZG is infinite, every winning parameter is eventually discovered
  by the reduced algorithm.
\end{theorem}

\begin{proof}
  Recall $\SymbState_0 = (\loc_0, \zone_{\mathrm{init}})$
  (all clocks zero, parameters unconstrained), and let
  $(\loc_0, \zone_0) = \SymbStateInit = \TempSucc{\SymbState_0}$ be the initial symbolic state
  of the canonical PZG.
  Let $(\loc_0, \zone'_0) = \repmap{\mergefun}(\loc_0, \zone_0)$ be its abstract
  counterpart.
  By the coverage requirement, $\zone_0 \subseteq \zone'_0$.
  By Lemma~\ref{lem:zone-corr} applied at $(\loc_0, \zone_0)$:
  \begin{equation*}
    \Win^*_\alpha[(\loc_0, \zone'_0)] \cap \zone_0 = \Win^*[(\loc_0, \zone_0)].
  \end{equation*}
  Since $\zone_{\mathrm{init}} \subseteq \zone_0$ (the zero-clock zone is contained in
  its time-successor closure):
  \begin{equation*}
    \Win^*_\alpha[(\loc_0, \zone'_0)] \cap \zone_{\mathrm{init}}
    = \Win^*[(\loc_0, \zone_0)] \cap \zone_{\mathrm{init}}.
  \end{equation*}
  Projecting onto parameters:
  \begin{equation*}
    \ProjectParam{\bigl(\Win^*_\alpha[(\loc_0, \zone'_0)] \cap \zone_{\mathrm{init}}\bigr)}
    = \ProjectParam{\bigl(\Win^*[(\loc_0, \zone_0)] \cap \zone_{\mathrm{init}}\bigr)}.
  \end{equation*}
  By the limit-correctness of the PTG algorithm on the canonical
  PZG~\cite{DBLP:conf/tacas/DahlsenJensenFPP24},
  the right-hand side equals the true winning parameter set.
  Since the left-hand side equals the right-hand side, the reduced algorithm
  converges to the same set.
\end{proof}

%% file: eval.tex
\label{sec:expe}
We evaluate the practical impact of the abstractions from \cref{sec:instantiations}, focusing on state-space reduction and scalability on previously intractable instances.

\subsection{Implementation}

The PTG synthesis algorithm of~\cite{DBLP:conf/tacas/DahlsenJensenFPP24} (\cref{alg:orig}) is implemented in \imitator~\cite{DBLP:conf/cav/Andre21}, a model checker supporting a wide range of PTA synthesis algorithms, with zones represented using the Parma Polyhedra Library~\cite{DBLP:journals/scp/BagnaraHZ08}.
We extend \imitator; the source code is available on GitHub.\footnote{\url{https://github.com/imitator-model-checker/imitator}, branch \texttt{develop}.}
Our extensions are threefold:

\textbf{Abstraction functions.}
The baseline without subsumption (\textsf{Baseline}) and zone inclusion (\textsf{Inc}) were already supported by~\cite{DBLP:conf/tacas/DahlsenJensenFPP24}.
We add double inclusion (\textsf{DInc}),
on-the-fly merge (\textsf{Mrg}), location-based abstraction (\textsf{Loc}),
and the convex (\textsf{CH}), octagonal (\textsf{Oct}) and box (\textsf{Box}) hull variants,
as described in \cref{sec:instantiations}.
We also implement two hybrid configurations (\textsf{OctH}, \textsf{BoxH})
that apply convex hull while zones remain below a constraint-count threshold,
falling back to octagonal or box hull when they grow more complex.

\textbf{Layer-based exploration with merge.}
We adopt the layer-based exploration strategy of~\cite{DBLP:conf/formats/AndreMPP22},
in which zones are explored layer by layer with abstraction applied between layers.

\textbf{Data-structure propagation.}
The main PTG-specific adaptation is propagating the game-theoretic data
structures across abstractions.
When a zone replaces another, the maps $\Win$ and $\Depends$ must be updated
by redirecting entries for the absorbed zone to its target,
combining values with zone union and set union respectively.

\subsection{Case Study: Adversarial IoT Scheduler}
\input{case_study.tex}

\subsection{Experimental Setup}
All experiments were conducted on a single core of a computer with an Intel Core
i5-1135G7 CPU @ 2.40GHz with 16 GB of RAM running Ubuntu 22.04.5 LTS.
Each configuration was run 5 times; we report the average time.
A timeout of 1{,}200\,s was applied per run.

We evaluate the abstractions described in \cref{sec:instantiations} using two case studies:
the adversarial IoT scheduler introduced in \cref{sec:casestudy}, and a production-cell benchmark~\cite{ProdCellCaseStudy}.
The production cell models a two-armed robot transferring plates between an input conveyor, a press, and an output conveyor.
The $n$-plate goal is for all plates to reach the output; the game is lost if two plates share the buffer.
A parameter sets the minimum wait between plates; robot timing is uncontrollable.

Each Malicious Synchronization instance is denoted~$(k/n)$, where $n$ is the number of devices and $k$ the congestion limit; Production Cell instances are numbered by plate count.
We compare the ten configurations introduced above, referring to each by its shorthand name.

\subsection{Results}

\begin{table}[t]
    \centering
    \caption{Experimental results. Each cell shows time in seconds followed by number of explored states. \timeout{} denotes timeout at 1{,}200\,s.}
    \label{tab:results}
    \resizebox{\textwidth}{!}{\input{table_results}}
\end{table}

The results are shown in \cref{tab:results}.
All abstractions consistently reduce the number of explored states
compared to \textsf{Baseline}, often by an order of magnitude.
All hull methods and \textsf{Loc} converge to identical state counts on several instances (e.g., 212 on~(2/2) and 1{,}438 on~(3/3)), suggesting they reach the same minimal reduced graph at those sizes.
\textsf{DInc} consistently improves over \textsf{Inc}, roughly halving the state count on most instances, but does not match the reductions achieved by \textsf{Mrg} or the hull methods.
Notably, our new abstractions solve instances intractable for the
prior technique (\textsf{Inc}): \textsf{Oct} is the only configuration to resolve
Malicious Synchronization~(3/4), doing so in 334\,s.

\paragraph{Convex and box hull.}
\textsf{CH} performs well on small instances
but can struggle on larger ones---though not monotonically:
it times out on Production Cell~4 yet solves instance~5 in 453\,s.
Empirically, convex hull abstractions can accumulate
increasingly complex polyhedra, at which point the underlying
arithmetic operations stall; whether this occurs appears to depend
on the geometry of the specific instance rather than its size alone.
\textsf{Box} avoids this by restricting to simpler shapes,
but its rectangular enclosure discards diagonal constraints
that are structurally important in the Production Cell's geometry,
which may lead to larger zones with more successors, potentially explaining the timeouts on all but the smallest instance.
\textsf{BoxH} can avoid both pitfalls: on Malicious Synchronization~(3/3),
it succeeds where \textsf{CH} and \textsf{Box} individually time out.

\paragraph{Octagonal hull.}
\textsf{Oct} is the most robust configuration:
it solves every instance solved by any other method
and is the only one to solve Malicious Synchronization~(3/4).
It appears to strike the right balance between aggressiveness
and geometric fidelity,
retaining enough diagonal structure to keep zone operations cheap
without over-approximating to the degree that causes the blowup seen with \textsf{Box}.
\textsf{OctH} behaves similarly to \textsf{Oct} across most instances,
with no consistent advantage in either direction.

\paragraph{Location-based abstraction.}
\textsf{Loc} can reach a state count comparable to the hull methods
yet be disproportionately slow:
on Malicious Synchronization~(3/3) it explores 1{,}438 states in 475\,s,
versus 24.3\,s for \textsf{Oct}.
The larger zones likely introduce spurious transitions
in the parametric zone graph,
causing repeated propagation through states
that do not contribute to convergence.

\paragraph{On-the-fly merge.}
\textsf{Mrg} is the strongest configuration that introduces
no geometric over-approximation of zones.
It handles all Production Cell instances comfortably
and is competitive on Malicious Synchronization,
making it a reliable fallback when a tighter zone representation
is preferred over raw speed.

Our evaluation is limited to two benchmark families; further case studies are needed to confirm that the relative performance of the abstractions generalises.

%% file: case_study.tex

\label{sec:casestudy}
We introduce a case study inspired by an adversarial IoT scenario in which a central scheduler coordinates multiple IoT devices under adversarial delays. We refer to this benchmark as \emph{Malicious Synchronization} in the experiments, reflecting the adversary's goal of forcing simultaneous transmissions. The setting is easily scalable in the number of devices and naturally induces the synchronized timing patterns that cause symbolic state explosion.

\subsubsection{System Overview}
The system consists of a network of battery-powered IoT devices that periodically transmit data to a cloud server, a central scheduler that grants transmission requests to prevent network congestion, and an adversary that can delay transmissions by dropping packets (modelled as a time delay). The adversary's goal is to force enough devices to transmit simultaneously, causing congestion. The adversary observes scheduler state but cannot interfere with scheduler-device communication, and has a limited energy budget bounding the number of delays.

\subsubsection{Parametric Timed Game Model}
The system is modelled as a network of PTG: one for each IoT device and one for the scheduler. Data transmission is simulated by enforcing a time delay in a \texttt{transmission} location of each device. Each device has an uncontrollable action, available when the adversary has remaining energy, that models adversarial delay by injecting additional time before transmission begins.

Three values are parametric: the minimum time between consecutive grants, the maximum adversarial delay per intervention, and the transmission duration. The adversary's energy budget is a discrete constant that bounds the total number of delay interventions. The goal is to reach a state where all devices have completed their data transmission, while avoiding states in which too many devices transmit simultaneously, causing congestion. 
An overview and the \imitator models for the two-device instance are given in \cref{app:case_study,app:case_study_model}.

%% file: table_results.tex
\newcolumntype{I}{!{\color{black!50}\vrule width 0.4pt}}
\newcolumntype{O}{!{\color{black}\vrule width 0.4pt}}
\newcommand{\TOo}{\multicolumn{2}{c!{\color{black}\vrule width 0.4pt}}{\timeout}}
\newcommand{\TOl}{\multicolumn{2}{c}{\timeout}}

\setlength{\tabcolsep}{2pt}
\begin{tabular}{
  l
  rIlO
  rIlO
  rIlO
  rIlO
  rIlO
  rIlO
  rIlO
  rIlO
  rIlO
  rIl
}
\toprule
 & \multicolumn{4}{c}{\textit{Prior work}} &
   \multicolumn{16}{c}{\textit{This paper}} \\
\cmidrule(r){2-5} \cmidrule(l){6-21}
\textbf{Instance} &
 \multicolumn{2}{c}{\textbf{Baseline}} &
 \multicolumn{2}{c}{\textbf{Inc}} &
 \multicolumn{2}{c}{\textbf{DInc}} &
 \multicolumn{2}{c}{\textbf{Mrg}} &
 \multicolumn{2}{c}{\textbf{Loc}} &
 \multicolumn{2}{c}{\textbf{CH}} &
 \multicolumn{2}{c}{\textbf{Oct}} &
 \multicolumn{2}{c}{\textbf{Box}} &
 \multicolumn{2}{c}{\textbf{OctH}} &
 \multicolumn{2}{c}{\textbf{BoxH}} \\
\midrule
\multicolumn{21}{l}{\textit{Malicious Synchronization}} \\[3pt]
(2/2) & 2.71 & 1319 & 1.33 & 561 & 1.01 & 318 & 0.93 & 274 & 1.29 & 212 & 0.85 & 212 & 0.74 & 212 & 0.66 & 212 & 0.81 & 212 & 0.83 & 212 \\
(2/3) & 32.9 & 11079 & 16.5 & 4428 & 9.76 & 2406 & 9.11 & 2158 & 10.5 & 1398 & 7.49 & 1398 & 5.96 & 1398 & 5.26 & 1398 & 7.46 & 1398 & 7.54 & 1398 \\
(3/3) & \TOo & 679 & 25694 & 574 & 14699 & 50.2 & 4416 & 475 & 1438 & \TOo & 24.3 & 1438 & \TOo & 47.5 & 1438 & 35.7 & 1438 \\
(2/4) & 485 & 94504 & 238 & 34896 & 125 & 17815 & 128 & 16484 & 64.6 & 8952 & \TOo & 49.8 & 8952 & 39.3 & 8952 & 61.4 & 8952 & 62.6 & 8952 \\
(3/4) & \TOo & \TOo & \TOo & \TOo & \TOo & \TOo & 334 & 9592 & \TOo & \TOo & \TOl \\
(4/4) & \TOo & \TOo & \TOo & \TOo & \TOo & \TOo & \TOo & \TOo & \TOo & \TOl \\
\midrule
\multicolumn{21}{l}{\textit{Production Cell}} \\[3pt]
1 & 0.34 & 771 & 0.07 & 77 & 0.07 & 73 & 0.06 & 59 & 0.05 & 91 & 0.05 & 55 & 0.06 & 55 & 0.06 & 55 & 0.05 & 55 & 0.05 & 55 \\
2 & \TOo & 0.51 & 507 & 0.43 & 424 & 0.31 & 306 & 1.10 & 544 & 0.29 & 222 & 0.43 & 222 & \TOo & 0.33 & 222 & 0.32 & 222 \\
3 & \TOo & 2.25 & 1576 & 1.19 & 913 & 0.60 & 476 & \TOo & 1.17 & 293 & 0.94 & 293 & \TOo & 0.70 & 293 & 0.96 & 293 \\
4 & \TOo & 23.2 & 7822 & 11.5 & 4054 & 2.31 & 973 & \TOo & \TOo & 3.91 & 684 & \TOo & 2.92 & 684 & \TOl \\
5 & \TOo & 111 & 26749 & 63.8 & 13704 & 6.14 & 2073 & \TOo & 453 & 1360 & 13.8 & 1360 & \TOo & 13.5 & 1360 & \TOl \\
\bottomrule
\end{tabular}

%% file: conclusion.tex

\label{sec:concl}

We presented a general abstraction framework for parametric timed games that reduces
the symbolic state space while preserving correctness of both parameter synthesis
and winning strategies. The framework subsumes the existing zone
inclusion technique for PTG and lifts double inclusion and zone merging from PTA and location-based abstraction from timed games to the PTG setting, while also introducing new hull abstractions inspired by classical numerical abstract domains.
Experiments on a production cell benchmark and a novel adversarial IoT case study show order-of-magnitude reductions in explored states, with octagonal hull emerging as the most robust configuration, solving every tractable instance, including instances previously intractable for existing techniques.

Several directions remain open.
First, combining abstractions, for instance merge steps followed by hull steps, could exploit the strengths of both.
Second, the hybrid configurations showed promise but relied on a fixed constraint-count threshold; a principled heuristic for switching from convex to a coarser hull would make these strategies more broadly applicable.
Finally, evaluating the framework on benchmarks with richer clock and parameter structure, would clarify which abstractions suit which problem characteristics.

%% file: appendix.tex

\section{Case Study Model}
\label{app:case_study_model}

\Cref{fig:scheduler,fig:devices} show the IMITATOR models for the two-device instance of the Malicious Synchronization case study. Urgent locations (U, yellow) disallow time passage; normal locations are green. Dotted edges are internal transitions, colored solid edges are synchronized actions, and dashed edges are uncontrollable (adversarial) actions.

\begin{figure}[H]
    \centering
    \includegraphics[width=0.95\linewidth]{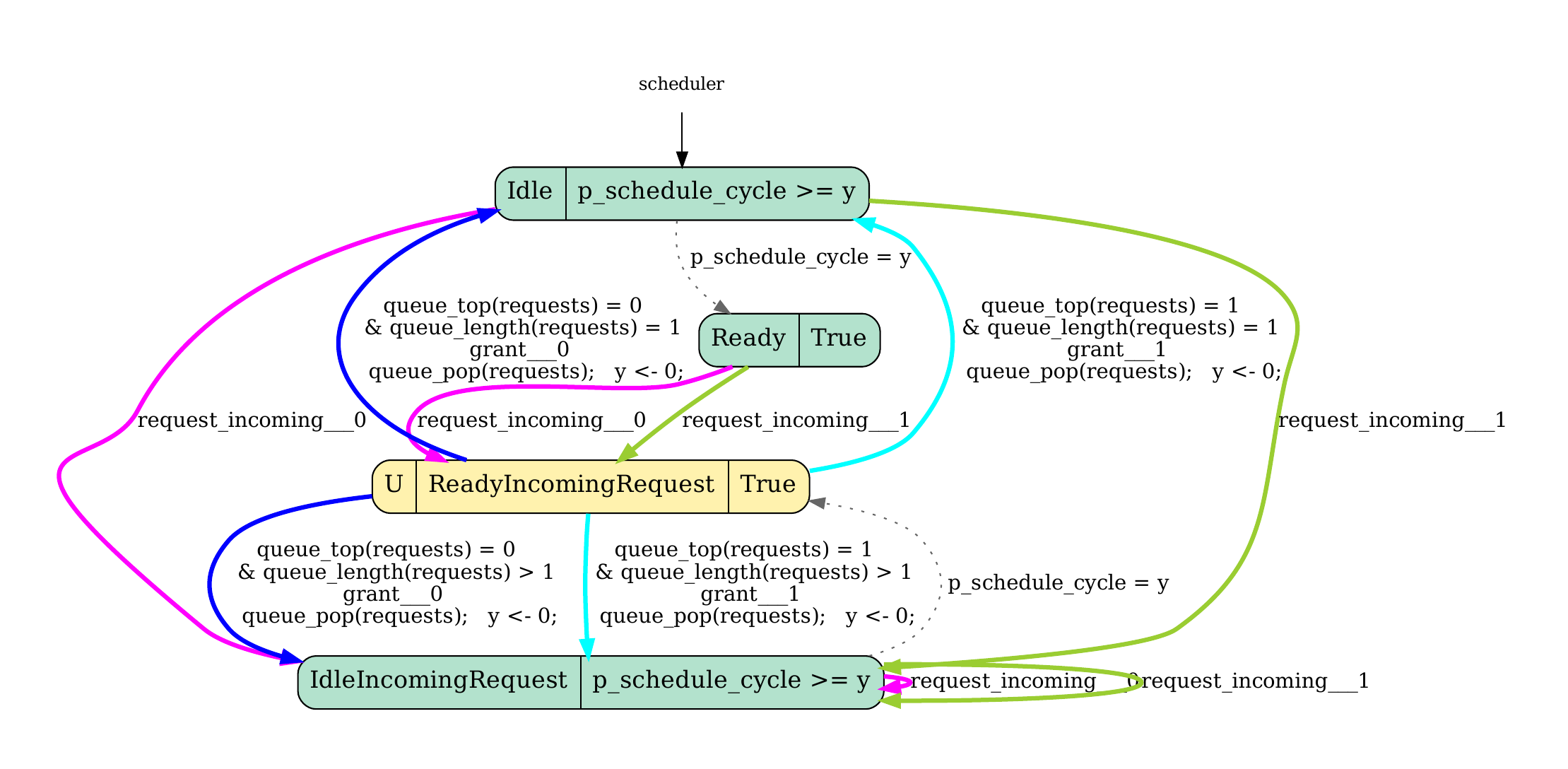}
    \caption{Scheduler automaton.\label{fig:scheduler}}
\end{figure}

\begin{figure}[H]
    \centering
    \includegraphics[width=0.95\linewidth]{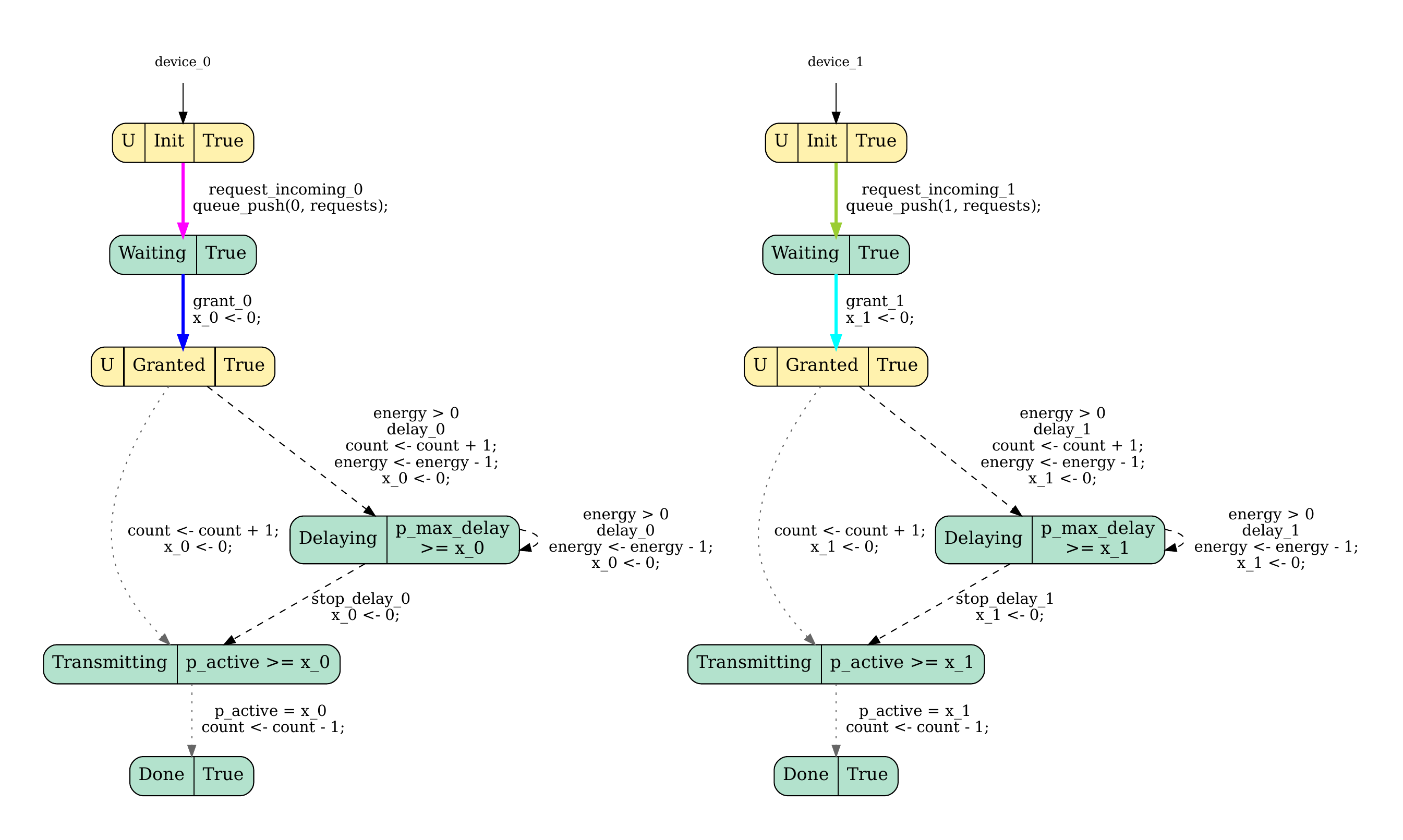}
    \caption{Device automata (device\_0 and device\_1).\label{fig:devices}}
\end{figure}
\newpage
\section{Case Study Overview}
\label{app:case_study}

\Cref{fig:case_study} gives a high-level overview of the adversarial IoT scheduling scenario described in \cref{sec:casestudy}.

\begin{figure}
    \centering
    \input{case_study_tikz.tex}
    \caption{Overview of the adversarial IoT scheduling scenario.\label{fig:case_study}}
\end{figure}
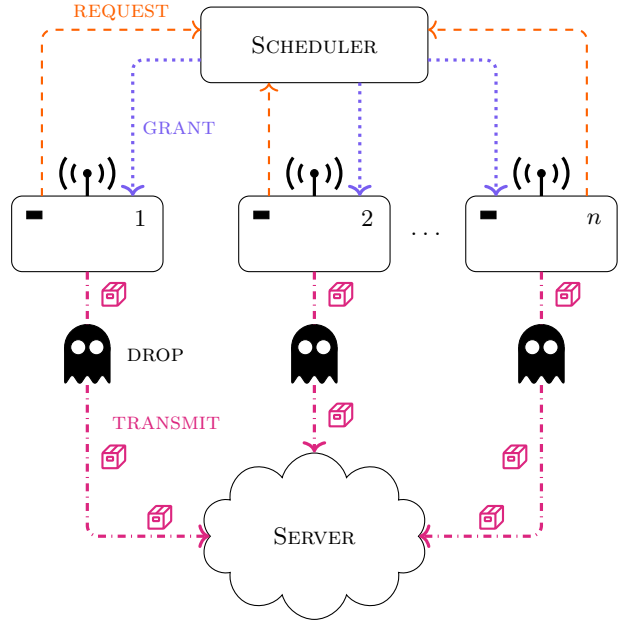

%% file: case_study_tikz.tex
\definecolor{ColorRequest}{HTML}{FE6100}
\definecolor{ColorGrant}{HTML}{785EF0}
\definecolor{ColorTransmit}{HTML}{DC267F}

\tikzset{
  device/.pic={
    \draw[rounded corners] (0,0) rectangle (2,1);
    \draw[fill=black] (.2,.8) rectangle (.4,.7);
    \draw[very thick] (1, 1) -- (1, 1.3);
    \draw[fill=black] (1,1.3) circle (.05);
    \begin{scope}
      \clip (1,1.3) -- (.6,1) -- (.6,1.6) -- (1,1.3)
         -- (1.4,1) -- (1.4,1.6) -- (1,1.3);
      \draw[very thick] (1,1.3) circle (.2);
      \draw[very thick] (1,1.3) circle (.35);
    \end{scope}
    \coordinate (-west) at (0,.5);
    \coordinate (-east) at (2,.5);
    \coordinate (-south) at (1,0);
    \coordinate (-nw) at (.4,1);
    \coordinate (-ne) at (1.6,1);
    \node at (1.7,.7) {\tikzpictext};
  },
  scheduler/.pic={
    \draw[rounded corners] (-1.5,-.5) rectangle (1.5,.5);
    \node at (0, 0) {\textsc{Scheduler}};
    \coordinate (-westhigh) at (-1.5,.2);
    \coordinate (-westlow) at (-1.5,-.2);
    \coordinate (-easthigh) at (1.5,.2);
    \coordinate (-eastlow) at (1.5,-.2);
    \coordinate (-southleft) at (-.6, -.5);
    \coordinate (-southright) at (.6, -.5);
  },
  adv/.pic={
    \draw[fill=black] 
    (0,0) circle (.3);
    \draw[rounded corners=.5mm, fill=black] 
    (-.3,0) -- 
    (-.3,-.5) -- (-.2,-.35) -- 
    (-.1,-.5) -- (0,-.35) -- 
    (.1,-.5) -- (.2,-.35) --
    (.3,-.5) -- (.3,0);
    \draw[fill=white](-.12,0) circle (.1)
    (.12,0) circle (.10);
    \coordinate (-north) at (0,.3); 
    \coordinate (-south) at (0,-.5);
  },
  data/.pic={
    \draw[thick,solid, rounded corners=.1mm, -] (0,.1) rectangle (.2,-.1);
    \draw[thick,solid, rounded corners=.1mm] (.06,.05) rectangle (.14,0);
    \draw[thick,solid, rounded corners=.1mm] 
    (0,.1) -- (0.1,.2) -- 
    (.3,.2) -- (0.2,.1);
    \draw[thick,solid, rounded corners=.1mm] 
    (.2,-.1) -- (.3,0) -- (.3,.2);
    \draw[thick,solid] (.1,.1) -- (.2,.2);
  }
}
\begin{tikzpicture}
  \tikzstyle{processed}=[preaction={fill, white}, pattern=crosshatch, pattern color=purple!70!blue, rounded corners];
  \tikzstyle{unprocessed}=[preaction={fill, white},pattern=north west lines, pattern color=blue];

  \tikzstyle{request}=[->,rounded corners, dashed, thick, color=ColorRequest];
  \tikzstyle{grant}=[->,rounded corners, dotted, very thick, color=ColorGrant];
  \tikzstyle{transmit}=[rounded corners, dashdotted, very thick, color=ColorTransmit];

  \draw pic[pic text=$1$] (d1) at (0,0) {device};
  \draw pic[pic text=$2$] (d2) at (3,0) {device};
  \draw pic[pic text=$n$] (d3) at (6,0) {device};
  \node at (5.5,.5) {\textsc{\dots}};
  \draw pic (sched) at (4,3) {scheduler};

  \draw[request] (d1-nw) |-  node [near end,above] {\textsc{request}}  (sched-westhigh) ;
  \draw[grant] (sched-westlow) -| node[near end, right] {\textsc{grant}} (d1-ne);

  \draw[request] (d2-nw) -- (sched-southleft);
  \draw[grant] (sched-southright) -- (d2-ne);
  
  \draw[request] (d3-ne) |-  (sched-easthigh) ;
  \draw[grant] (sched-eastlow) -| (d3-nw);

  \draw[scale=1, transform shape] pic (adv1) at (1, -1) {adv};
  \draw[scale=1, transform shape] pic (adv2) at (4, -1) {adv};
  \draw[scale=1, transform shape] pic (adv3) at (7, -1) {adv};

  \draw[transmit] (d1-south) -- pic[right=2mm] {data} (adv1-north);
  \draw[transmit] (d2-south) -- pic[right=2mm] {data} (adv2-north);
  \draw[transmit] (d3-south) -- pic[right=2mm] {data} (adv3-north);

  \node[cloud, draw, minimum width = 3cm,
    minimum height = 2cm] (cloud) at (4,-3.5) {\textsc{Server}};

  \draw[transmit,->] (adv1-south) |- node [very near start, right=2mm] {\textsc{transmit}} pic[near start, right=2mm, -] {data}
  pic[near end, above=2mm, -] {data} (cloud.west);
  \draw[transmit,->] (adv2-south) -- pic[right=2mm, -] {data} (cloud.north);
  \draw[transmit,->] (adv3-south) |- pic[near start, left=5mm, -] {data}
  pic[near end, above=2mm, -] {data}  (cloud.east);
  
  \node at (1.9,-1.1) {\textsc{drop}};

\end{tikzpicture}

%% file: main.bbl
\begin{thebibliography}{10}

\bibitem{Intro-TA}
Rajeev Alur and David~L. Dill.
\newblock A theory of timed automata.
\newblock {\em Theoretical Computer Science}, 126(2):183--235, 1994.

\bibitem{DBLP:conf/stoc/AlurHV93}
Rajeev Alur, Thomas~A. Henzinger, and Moshe~Y. Vardi.
\newblock Parametric real-time reasoning.
\newblock In {\em {STOC}}, pages 592--601. {ACM}, 1993.

\bibitem{DBLP:conf/cav/Andre21}
{\'{E}}tienne Andr{\'{e}}.
\newblock {IMITATOR} 3: Synthesis of timing parameters beyond decidability.
\newblock In {\em {CAV} {(1)}}, volume 12759 of {\em Lecture Notes in Computer
  Science}, pages 552--565. Springer, 2021.

\bibitem{DBLP:conf/atva/AndreFS13}
{\'{E}}tienne Andr{\'{e}}, Laurent Fribourg, and Romain Soulat.
\newblock Merge and conquer: State merging in parametric timed automata.
\newblock In {\em {ATVA}}, Lecture Notes in Computer Science, pages 381--396.
  Springer, 2013.

\bibitem{DBLP:conf/rp/AndreLR15}
{\'{E}}tienne Andr{\'{e}}, Didier Lime, and Olivier~H. Roux.
\newblock Integer-complete synthesis for bounded parametric timed automata.
\newblock In {\em {RP}}, volume 9328 of {\em Lecture Notes in Computer
  Science}, pages 7--19. Springer, 2015.

\bibitem{DBLP:conf/formats/AndreMPP22}
{\'{E}}tienne Andr{\'{e}}, Dylan Marinho, Laure Petrucci, and Jaco van~de Pol.
\newblock Efficient convex zone merging in parametric timed automata.
\newblock In {\em {FORMATS}}, volume 13465 of {\em Lecture Notes in Computer
  Science}, pages 200--218. Springer, 2022.

\bibitem{DBLP:conf/nfm/ArcileA22}
Johan Arcile and {\'{E}}tienne Andr{\'{e}}.
\newblock Zone extrapolations in parametric timed automata.
\newblock In {\em {NFM}}, Lecture Notes in Computer Science, pages 451--469.
  Springer, 2022.

\bibitem{DBLP:journals/scp/BagnaraHZ08}
Roberto Bagnara, Patricia~M. Hill, and Enea Zaffanella.
\newblock The parma polyhedra library: Toward a complete set of numerical
  abstractions for the analysis and verification of hardware and software
  systems.
\newblock {\em Sci. Comput. Program.}, 72(1-2):3--21, 2008.

\bibitem{DBLP:conf/rtss/Balarin96}
Felice Balarin.
\newblock Approximate reachability analysis of timed automata.
\newblock In {\em {RTSS}}, pages 52--61. {IEEE} Computer Society, 1996.

\bibitem{DBLP:conf/tacas/BehrmannBLP04}
Gerd Behrmann, Patricia Bouyer, Kim~Guldstrand Larsen, and Radek Pel{\'{a}}nek.
\newblock Lower and upper bounds in zone based abstractions of timed automata.
\newblock In {\em {TACAS}}, Lecture Notes in Computer Science, pages 312--326.
  Springer, 2004.

\bibitem{DBLP:conf/sefm/BezdekBBC16}
Peter Bezdek, Nikola Benes, Jiri Barnat, and Ivana Cern{\'{a}}.
\newblock {LTL} parameter synthesis of parametric timed automata.
\newblock In {\em {SEFM}}, Lecture Notes in Computer Science, pages 172--187.
  Springer, 2016.

\bibitem{DBLP:conf/concur/CassezDFLL05}
Franck Cassez, Alexandre David, Emmanuel Fleury, Kim~Guldstrand Larsen, and
  Didier Lime.
\newblock Efficient on-the-fly algorithms for the analysis of timed games.
\newblock In {\em {CONCUR}}, Lecture Notes in Computer Science, pages 66--80.
  Springer, 2005.

\bibitem{DBLP:conf/popl/CousotC77}
Patrick Cousot and Radhia Cousot.
\newblock Abstract interpretation: {A} unified lattice model for static
  analysis of programs by construction or approximation of fixpoints.
\newblock In {\em {POPL}}, pages 238--252. {ACM}, 1977.

\bibitem{DBLP:conf/popl/CousotH78}
Patrick Cousot and Nicolas Halbwachs.
\newblock Automatic discovery of linear restraints among variables of a
  program.
\newblock In {\em {POPL}}, pages 84--96. {ACM} Press, 1978.

\bibitem{abstractions_artifact}
Mikael~B. Dahlsen-Jensen, Laure Petrucci, and J.C. van~de Pol.
\newblock Artifact for "state-space abstractions for parametric timed games".
\newblock Zenodo,
  \href{https://doi.org/10.5281/zenodo.19680181}{10.5281/zenodo.19680181},
  April 2026.

\bibitem{DBLP:conf/tacas/DahlsenJensenFPP24}
Mikael~Bisgaard Dahlsen{-}Jensen, Baptiste Fievet, Laure Petrucci, and Jaco
  van~de Pol.
\newblock On-the-fly algorithm for reachability in parametric timed games.
\newblock In {\em {TACAS} {(3)}}, Lecture Notes in Computer Science, pages
  194--212. Springer, 2024.

\bibitem{DBLP:conf/qestformats/DahlsenJensenFPP25}
Mikael~Bisgaard Dahlsen{-}Jensen, Baptiste Fievet, Laure Petrucci, and Jaco
  van~de Pol.
\newblock Controller synthesis for parametric timed games.
\newblock In {\em {QEST+FORMATS}}, Lecture Notes in Computer Science, pages
  314--332. Springer, 2025.

\bibitem{DBLP:conf/lics/AlfaroGJ04}
Luca de~Alfaro, Patrice Godefroid, and Radha Jagadeesan.
\newblock Three-valued abstractions of games: Uncertainty, but with precision.
\newblock In {\em {LICS}}, pages 170--179. {IEEE} Computer Society, 2004.

\bibitem{DBLP:conf/avmfss/Dill89}
David~L. Dill.
\newblock Timing assumptions and verification of finite-state concurrent
  systems.
\newblock In {\em Automatic Verification Methods for Finite State Systems},
  Lecture Notes in Computer Science, pages 197--212. Springer, 1989.

\bibitem{DBLP:conf/sas/HenzingerMMR00}
Thomas~A. Henzinger, Rupak Majumdar, Freddy Y.~C. Mang, and
  Jean{-}Fran{\c{c}}ois Raskin.
\newblock Abstract interpretation of game properties.
\newblock In {\em {SAS}}, Lecture Notes in Computer Science, pages 220--239.
  Springer, 2000.

\bibitem{DBLP:conf/tacas/HuneRSV01}
Thomas Hune, Judi Romijn, Mari{\"{e}}lle Stoelinga, and Frits~W. Vaandrager.
\newblock Linear parametric model checking of timed automata.
\newblock In {\em {TACAS}}, Lecture Notes in Computer Science, pages 189--203.
  Springer, 2001.

\bibitem{DBLP:conf/concur/JensenLLS25}
Nicolaj~{\O}. Jensen, Kim~G. Larsen, Didier Lime, and Jir{\'{\i}} Srba.
\newblock On-the-fly symbolic algorithm for timed {ATL} with abstractions.
\newblock In {\em {CONCUR}}, volume 348 of {\em LIPIcs}, pages 25:1--25:19.
  Schloss Dagstuhl - Leibniz-Zentrum f{\"{u}}r Informatik, 2025.

\bibitem{DBLP:conf/wodes/JovanovicFLR12}
Aleksandra Jovanovic, S{\'{e}}bastien Faucou, Didier Lime, and Olivier~H. Roux.
\newblock Real-time control with parametric timed reachability games.
\newblock In {\em {IFAC WODES}}, pages 323--330. Elseviers, 2012.

\bibitem{DBLP:journals/ijcon/JovanovicLR19}
Aleksandra Jovanovic, Didier Lime, and Olivier~H. Roux.
\newblock A game approach to the parametric control of real-time systems.
\newblock {\em Int. J. Control}, 92(9):2025--2036, 2019.

\bibitem{DBLP:journals/sttt/LarsenPY97}
Kim~Guldstrand Larsen, Paul Pettersson, and Wang Yi.
\newblock {UPPAAL} in a nutshell.
\newblock {\em Int. J. Softw. Tools Technol. Transf.}, 1(1-2):134--152, 1997.

\bibitem{ProdCellCaseStudy}
Claus Lewerentz and Thomas Lindner.
\newblock {\em Case study ``production cell'': A comparative study in formal
  specification and verification}, pages 388--416.
\newblock Springer Berlin Heidelberg, Berlin, Heidelberg, 1995.

\bibitem{Classic-TG}
Oded Maler, Amir Pnueli, and Joseph Sifakis.
\newblock On the synthesis of discrete controllers for timed systems.
\newblock In Ernst~W. Mayr and Claude Puech, editors, {\em STACS 95}, pages
  229--242, Berlin, Heidelberg, 1995. Springer Berlin Heidelberg.

\bibitem{DBLP:conf/wcre/Mine01}
Antoine Min{\'{e}}.
\newblock The octagon abstract domain.
\newblock In {\em {WCRE}}, page 310. {IEEE} Computer Society, 2001.

\bibitem{DBLP:conf/cav/PeterEM11}
Hans{-}J{\"{o}}rg Peter, R{\"{u}}diger Ehlers, and Robert Mattm{\"{u}}ller.
\newblock Synthia: Verification and synthesis for timed automata.
\newblock In {\em {CAV}}, Lecture Notes in Computer Science, pages 649--655.
  Springer, 2011.

\end{thebibliography}
